%% file: MPComm_R1.tex
\documentclass[journal]{IEEEtran}
\usepackage{amsmath,amsfonts}
\usepackage{algorithmic}
\usepackage{algorithm}
\usepackage{array}
\usepackage[caption=false,font=footnotesize,labelfont=sf,textfont=sf]{subfig}
\usepackage{textcomp}
\usepackage{stfloats}
\usepackage{url}
\usepackage{verbatim}
\usepackage{graphicx}
\usepackage{cite}
\usepackage[acronym]{glossaries}
\newcommand{\newac}{\newacronym}
\newcommand{\ac}{\gls}

\newcommand{\acpl}{\glspl}
\newcommand{\Acpl}{\Glspl}
\input{bib_files/BL_acronym}

\usepackage{xcolor}
\usepackage[table]{xcolor}
\definecolor{lightcyan}{rgb}{0.88, 1.0, 1.0}
\usepackage{mathrsfs,amsmath,amssymb,amsfonts,mathtools,bm}
\usepackage{tikz, pgfplots}
\usetikzlibrary{shapes, arrows, arrows.meta, fit, positioning, calc, 3d}
\usepgfplotslibrary{fillbetween}
\pgfplotsset{compat=1.18} 
\tikzset{>=latex}

\newtheorem{theorem}{Theorem}
\newtheorem{lemma}{Lemma}

\newtheorem{definition}{Definition}
\newtheorem{proposition}{Proposition}

\newtheorem{corollary}{Corollary}
 
\begin{document}
\bstctlcite{IEEEexample:BSTcontrol}

\title{Model Predictive Communication for Timely Status Updates in Low-Altitude Networks}
\author{
\thanks{Part of this work will be presented at the IEEE ICC 2026.}
\thanks{This work has been supported by ELLIIT and 6G-LEADER, 101192080.}
\IEEEauthorblockN{Bowen~Li, Jiping~Luo, Themistoklis~Charalambous, Nikolaos~Pappas}
\thanks{B. Li, J. Luo, and N. Pappas are with the Department of Computer and Information Science, Link{\"o}ping University, 58183 Link{\"o}ping, Sweden. Emails: \{bowen.li, jiping.luo, nikolaos.pappas\}@liu.se}
\thanks{T. Charalambous is with the Department of Electrical and Computer Engineering, University of Cyprus, 1678 Nicosia, Cyprus. Email: charalambous.themistoklis@ucy.ac.cy} 
}
\maketitle
\begin{abstract}
Timely information delivery in low-altitude networks is critical for many time-sensitive applications, such as unmanned aerial vehicle (UAV) navigation, inspection, and surveillance. The key challenge lies in balancing three competing factors: stringent data freshness requirements, UAV onboard energy consumption, and interference with terrestrial services. Addressing this challenge requires not only efficient power and channel allocation strategies but also effective communication timing over the entire operation horizon.
In this work, we propose a model predictive communication (MPComm) framework, which leverages predictive channel statistics to control both data generation and transmission.
Within this framework, we formulate a constrained bi-objective optimization problem to achieve a desired trade-off between energy consumption and terrestrial channel occupation, subject to a strict timeliness constraint. We solve this problem using Pareto analysis and show that the original non-convex, mixed-integer problem can be decomposed into a two-layer structure: the outer layer determines the optimal communication timing, while the inner layer determines the optimal power and channel allocation for each communication interval. An efficient algorithm for the inner problem is developed using non-convex analysis, with asymptotic optimality guarantees, while the outer problem is solved optimally via a simple graph search, with edges characterized by inner solutions. The proposed approach applies to a broad class of problem variants, including objective transformations and single-objective specializations.
Numerical results verify near-optimality and show performance gains exceeding 20 dB from long-horizon prediction and 6 dB from joint generation-transmission control, with a nearly 100\% success rate for time-sensitive services.
\end{abstract}


\section{Introduction\label{sec:intro}}
Low-altitude activities have grown significantly over the past decade, leading to a surge in demand for timely communications that support time-critical applications such as real-time navigation, inspection, and surveillance~\cite{WuXuZenNg:J21}. In such applications, stringent data freshness requirements are typically imposed~\cite{SonLinWanSun:M25,AbdMar:M26}. For instance, timely access to \ac{uav} sensory data is essential for aerial traffic coordination, mission monitoring, and decision-making at the control center; obsolete information can degrade operational safety and even lead to crashes.

Most existing works on aerial communications focus on the throughput-reliability-delay trade-off~\cite{MatShe:J25,ZhaLiRonZen:J25}. However, such approaches may not suffice for many safety-critical applications. \Acpl{uav} typically face strict energy constraints and limited communication budgets. It is inefficient, if not impossible, to communicate every piece of sensory data to the control center. Instead, the most recent and relevant data should be generated and delivered promptly~\cite{KauYatGru:C12}. A natural question then arises: \emph{What is the optimal timing of communication?} The answer is not obvious, especially in low-altitude networks where channel conditions and network topology can change rapidly. Recent studies~\cite{LuoPap:J25b, luo2026role, luo2025information} have investigated this problem from information-theoretic and control-theoretic perspectives. Intuitively, transmitting immediately, given abundant communication resources, may be a bad timing if better channel conditions are expected in the near future. In this article, we employ the \ac{aoi} metric~\cite{KauYatGru:C12} as an indicator of data freshness and address this question by proposing a predictive communication strategy.

Low-altitude networks, while dynamic, are often predictable~\cite{CheLiSunCui:M26}. First, in many applications such as inspection and cargo delivery~\cite{Mei:23,Win:23,Aon:23}, \acpl{uav} follow predetermined trajectories and their future movements are largely predictable~\cite{HeLiMoHua:J24}. Second, advanced channel sensing techniques, such as radio maps and digital twins~\cite{ZenCheXuWu:J24,WanZhaNieYu:M25,LiZhaJiaChe:25,WanZhaZhaSun:J25}, provide predictable and high-precision 3D representations of the wireless propagation environment. In such applications, the communication model can be predicted with reasonable accuracy.
This enables \emph{model predictive communication} (MPComm), where resource management, data acquisition, and access control are optimized over a longer planning horizon.

There are several key trade-offs in the design of low-altitude, timely communication networks. The first trade-off is between \ac{uav} energy consumption and terrestrial service quality. The coexistence of aerial and terrestrial traffic is complicated by dominant \ac{los} channels, which induce severe cross-network interference~\cite{MeiZha:M21,VaeLinZhaSad:M24}. Consequently, improving aerial communication throughput generally requires either more onboard energy or more terrestrial spectrum resources. The system therefore trades \ac{uav} energy consumption against the terrestrial service quality. The second trade-off is between resource utilization cost and data freshness. To meet stringent timeliness constraints, appropriate sample/communication timing and resource allocation decisions must be made.

Many existing works focus on \ac{uav} trajectory optimization and resource allocation to minimize AoI in data collection tasks~\cite{HuXioQuNi:J21,LiuTonWanBai:J21,QinWeiQuZho:J23,YuaCheHeHou:J25,TanLiuHeXie:J25,LuWuJiaFei:J26}. However, most of them adopt packet-level metrics that measure the lifetime of each update, which is insufficient for time-critical tasks such as \ac{uav} status monitoring. Another line of research adopts source-level metrics to characterize the freshness of the information at the remote center~\cite{GirParBenDeb:J21,SonShaCheZhu:J24,SunLvTanWan:J24,ZhaPapZhaYan:J25,MuLuJiaChe:J26}. Nevertheless, these works do not exploit the predictability of low-altitude networks and therefore remain essentially reactive rather than proactive. Recent studies on predictive communications leverage predictive information to optimize routing, resource allocation, beamforming, and local negotiation in low-altitude \ac{uav} networks~\cite{LiChe:J24a,LiChe:J24b,LiSuSuPen:J25,DuWeiYanYan:J26,LiCheLiuXu:J26,LiChe:J26}, but treat data generation as exogenous, potentially forcing generation and transmission under unfavorable channel conditions and requiring more resources.

In this work, we develop an MPComm framework that exploits predicted channel statistics over the mission horizon to jointly optimize data generation and transmission for timely status updates.\footnote{
MPComm is conceptually related to model predictive control (MPC): both optimize decisions over a finite horizon using predictive models. However, MPComm is a specific data-generation-and-transmission optimization framework, whereas MPC is a general receding-horizon control paradigm. With online state and prediction updates, MPComm can serve as the per-horizon data-generation and transmission optimizer within an MPC implementation.}
Specifically, we formulate a bi-objective optimization problem to achieve a desired trade-off between aerial energy consumption and terrestrial spectrum usage, subject to a hard timeliness constraint on aerial traffic. The decision variables include the timing of status updates, the \ac{uav}'s power allocation, and the occupation of terrestrial spectrum. These decisions are optimized over a long planning horizon to adapt to predicted channel conditions. This constrained bi-objective problem poses two main challenges: (i) it is a non-convex, mixed-integer program, and (ii) its complete solution constitutes a Pareto frontier, which is computationally challenging to obtain. We propose a low-complexity method to address these challenges.

Our main contributions are summarized below.

(1) We propose an MPComm framework for timely status updates in low-altitude networks. Within this framework, we formulate a constrained bi-objective optimization problem to satisfy data freshness requirements while balancing energy consumption and terrestrial channel usage. Unlike conventional designs that focus solely on data transmission, the proposed framework directly controls information generation and transmission through an age-aware sampling controller.

(2) We exploit the problem structure to characterize the Pareto frontier and derive tight upper and lower bounds for its region. These results significantly reduce computation overhead. 
Building on this characterization, we establish transfer rules that recover the Pareto frontier or optimal solutions of transformed-objective, preference-based, and budget-constrained variants directly from the original frontier. These rules avoid solving potentially nonconvex or intractable variants, enabling efficient adaptation to changing design preferences and resource constraints.

(3) We develop an efficient algorithm to solve the bi-objective problem. The problem is first decomposed into a two-layer structure, where the outer layer determines the optimal communication timing, while the inner layer determines the optimal power and channel allocation for each communication interval. Leveraging this decomposition and non-convex analysis, we propose a graph-based algorithm with complexity linear in the number of \acpl{bs} and the planning horizon, and quadratic in the number of channel \acpl{rb}. Numerical results verify the near-optimality of the proposed algorithms and separately quantify the gains: over $20$ dB of performance improvement from long-horizon predictive control, approximately $6$ dB from joint data-generation and transmission control, and a nearly $100\%$ successful-service rate under per-update timeliness guarantees.

The rest of the paper is organized as follows. Section~\ref{sec:System-Model} introduces the MPComm framework. Section~\ref{sec:decomposition} characterizes the Pareto frontier and decomposes the optimization problem. Section~\ref{sec:algorithm} develops a graph-based algorithm. Section~\ref{sec:variants} introduces the variant problems and discusses their connections to the original problem. Numerical results are provided in Section~\ref{sec:Simulation}, and we conclude this paper in Section~\ref{sec:Conclusion}.

\emph{Notation:} Let $\mathbb{R}$, $\mathbb{R}_+$, $\mathbb{Z}$, and $\mathbb{Z}_+$ denote the sets of real numbers, nonnegative real numbers, integers, and nonnegative integers, respectively. Calligraphic letters denote sets or index sets, {\em e.g.}, $\mathcal{N}\triangleq\{1,\ldots,N\}$, $\mathcal{K}\triangleq\{1,\ldots,K\}$, and $\mathcal{T}\triangleq\{1,\ldots,T\}$. $|\mathcal{N}|$ denotes the cardinality of $\mathcal{N}$, and $\mathcal{N}^{K}$ denotes its $K$-fold Cartesian product. Bold lowercase letters, bold uppercase letters, and bold calligraphic letters denote vectors, matrices, and tensors, respectively. For example, for a third-order tensor $\boldsymbol{\mathcal{X}}\in\mathbb{R}^{N\times K\times T}$ with entries $x_n[k,t]$, we write
$\boldsymbol{\mathcal{X}}=\{x_n[k,t]\}_{n\in\mathcal{N},k\in\mathcal{K},t\in\mathcal{T}}$,
$\mathbf{X}_n=\{x_n[k,t]\}_{k\in\mathcal{K},t\in\mathcal{T}}\in\mathbb{R}^{K\times T}$,
$\mathbf{X}[t]=\{x_n[k,t]\}_{n\in\mathcal{N},k\in\mathcal{K}}\in\mathbb{R}^{N\times K}$,
$\mathbf{X}[k]=\{x_n[k,t]\}_{n\in\mathcal{N},t\in\mathcal{T}}\in\mathbb{R}^{N\times T}$, and
$\mathbf{x}_{n,k}=\{x_n[k,t]\}_{t\in\mathcal{T}}\in\mathbb{R}^{T}$.
For $x,a\in\mathbb{R}$, $x\to a^-$ and $x\to a^+$ denote that $x$ approaches $a$ from the left and from the right, respectively.

\section{System Model\label{sec:System-Model}}
Consider a \ac{uav} telemetry system, as depicted in Fig.~\ref{fig:system_model}.
The system consists of one \ac{uav} indexed by $0$, and $N$ \acpl{bs} indexed by $n\in\mathcal{N}=\{1,\cdots,N\}$. The \ac{uav} updates its on-board sensory data to a control center via the terrestrial network ({\em i.e.}, the \acpl{bs}). Timely delivery of this information is crucial, as the control center relies on the most recent and relevant data for reliable analysis and informed decision-making. Meanwhile, the terrestrial network shall also maintain stable service for the ground users. The MPComm framework is detailed below.

\begin{figure}
\begin{centering}
\includegraphics[width=1\columnwidth]{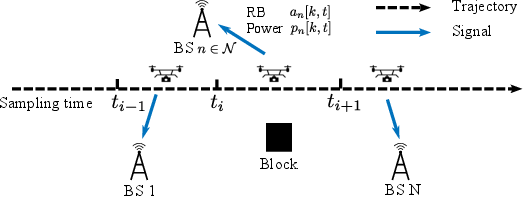}
\par\end{centering}
\caption{Illustration of the \ac{uav} telemetry system. The \ac{uav} symbols along the trajectory illustrate the \ac{uav}'s positions at different time instants. The \ac{uav} reports on-board data to a control center through \acpl{bs}.}
\label{fig:system_model}
\end{figure}

\subsection{Predictive Channel Model}
We consider a slotted orthogonal frequency division multiplexing system, where time is divided into
slots indexed by $t\in\mathcal{T}=\{1,\cdots,T\}$, and the available spectrum is partitioned into $K$ orthogonal \acpl{rb}, indexed by $k\in\mathcal{K}=\{1,\cdots,K\}$.
The wireless channel from the \ac{uav} to \ac{bs} $n$ on \ac{rb} $k$ at time $t$ can be modeled as
\begin{equation}
h_{n}\left[k,t\right]=g_{n}\left[k,t\right]\underbrace{\zeta^\mathrm{pe}_n\left[k,t\right]\zeta^\mathrm{sf}_{n}\left[k,t\right]}_{\triangleq\zeta_n},\label{eq:channel_model}
\end{equation}
where $g_{n}[k,t]$ is the large-scale channel gain (\emph{e.g.}, path loss and shadowing), {$\zeta^\mathrm{pe}_n[k,t]$ is the channel prediction error (\emph{e.g.}, radio map error) following the Lognormal distribution with unit mean~\cite{LevYapKutCai:J21}, \emph{i.e.}, $\zeta^\mathrm{pe}_{n}[k,t]\sim\text{Lognormal}(-\sigma^2_{n}[k,t]/2,\sigma^2_{n}[k,t])$,} and $\zeta^\mathrm{sf}_{n}[k,t]$ is the small-scale fading following the Gamma distribution with unit mean~\cite{LiChe:J24b,Nak:B60}, \emph{i.e.}, $\zeta^\mathrm{sf}_{n}[k,t]\sim\text{Gamma}(\kappa_{n}[k,t],1/\kappa_{n}[k,t])$, where $\kappa_{n}[k,t]$ denotes the shape parameter and $1/\kappa_{n}[k,t]$ denotes scale parameter.
{Here, $\sigma_n^2[k,t]\geq0$ quantifies radio-map prediction uncertainty, with $\sigma_n^2[k,t]=0$ indicating error-free prediction, while $\kappa_n[k,t]\geq1$ quantifies \ac{los} strength, with $\kappa_n[k,t]\to\infty$ corresponding to a deterministic channel.}

The predictive channel model is built on two key enablers.
\begin{itemize}
\item advanced channel sensing techniques, such as radio maps and digital
twins~\cite{ZenCheXuWu:J24,WanZhaNieYu:M25,LiZhaJiaChe:25,WanZhaZhaSun:J25}, which provide a 3D representation of the wireless propagation environment and offer spatially resolved channel statistics between the \ac{uav} and ground \acpl{bs}; and
\item high-precision \ac{uav} control, which allows the \ac{uav} to follow pre-determined trajectories $\{(t,\mathbf{p}_{0}[t])\}_{t\in\mathcal{T}}$ with minimal deviation~\cite{HeLiMoHua:J24,CheLiSunCui:M26}.
\end{itemize}

Consequently, the \ac{uav}'s motion traces a one-dimensional slice through the 3D channel field, yielding a time-indexed channel profile that can be predicted in advance. Let
\begin{equation*}
\left\{g_{n}\left[k,t\right],\sigma_{n}\left[k,t\right],\kappa_{n}\left[k,t\right]\right\}_{k\in \mathcal{K}}\triangleq\Xi\left(\mathbf{p}_{0}\left[t\right],\mathbf{p}_{n}\right),\,\forall t\in\mathcal{T}
\end{equation*}
denote the radio map between the \ac{uav} and \ac{bs} $n$ with position $\mathbf{p}_{n}$ along the trajectory. The predictive channel information is available to the control center prior to deployment.

\subsection{Communication Model}
Denote by $a_{n}[k,t]\in\{0,1\}$ the indicator of allocating \ac{rb} $k$ at slot $t$ for the transmission from the \ac{uav} to \ac{bs} $n$. For each \ac{rb} $(k,t)$, the transmitting \ac{uav} can be scheduled to at most one \ac{bs}, {\em i.e.}, $\sum_{n\in\mathcal{N}}a_{n}\left[k,t\right]\le1.$ Accordingly, the \ac{rb} scheduling policy $\mathbf{A}[t]$ for any time $t\in\mathcal{T}$ is restricted to the feasible set
\begin{equation}
\mathcal{S}_{\mathbf{A}} \hspace{-0.2em}\triangleq \hspace{-0.2em}
\Big\{\mathbf{A} \in 
\left\{ 0,1\right\} ^{\left|\mathcal{N}\right|\times\left|\mathcal{K}\right|}\hspace{-0.2em}:\hspace{-0.2em}
\sum_{n\in\mathcal{N}}a_{n}\left[k,t\right]\le1,\forall k\in\mathcal{K}\Big\} .\label{eq:def_fs_A}
\end{equation}

Let $p_{n}\left[k,t\right] \ge 0$ denote the transmit power of the \ac{uav} associated
with $\left(n,k,t\right)$. The total transmission power is limited
to the threshold $\bar{p}$, leading to the sum-power constraint $\mathbf{P}\left[t\right]\in\mathcal{S}_{\mathbf{P}}$,
where $\mathcal{S}_{\mathbf{P}}$ denotes the feasible power set
\begin{equation}
\mathcal{S}_{\mathbf{P}}\triangleq\Big\{ \mathbf{P}\in\mathbb{R}_{+}^{\left|\mathcal{N}\right|\times\left|\mathcal{K}\right|}:\sum_{n\in\mathcal{N}}\sum_{k\in\mathcal{K}}p_{n}\left[k,t\right]\le\bar{p}\Big\} .\label{eq:def_fs_P}
\end{equation}

We now derive the sum data rate. The \ac{snr} for the link from the \ac{uav} to \ac{bs} $n$ on \ac{rb} $k$ at time $t$ is given by 
\begin{equation}
\gamma_{n}\left[k,t\right]=\frac{p_{n}\left[k,t\right]h_{n}\left[k,t\right]}{\delta^{2}},\label{eq:def_snr}
\end{equation}
where $\delta^{2}$ is the noise power. We assume perfect Doppler compensation~\cite{LuZen:J24}. Then, the instantaneous channel capacity from the \ac{uav} to \ac{bs} $n$ at time $t$ for \acpl{rb} $k$ is modeled as 
\begin{equation}
c_{n}\left[k,t\right]=\log_{2}\left(1+\gamma_{n}\left[k,t\right]\right).\label{eq:def_c}
\end{equation}

Finally, the sum data rate over a time interval $(t^{\prime},t^{\prime\prime})$ aggregated across all \acpl{bs} and \acpl{rb} is given by\footnote{We assume reliable wired connectivity among the \acpl{bs}/control center, so that the data received by different \acpl{bs} can be aggregated cooperatively. A non-cooperative variant can be modeled by requiring that each update interval be delivered through only one \ac{bs}. This variant preserves the problem structure and can be handled by the proposed algorithmic framework.}
\begin{equation}
\upsilon\left(t^{\prime},t^{\prime\prime}\right)=\sum_{n\in\mathcal{N}}\sum_{k\in\mathcal{K}}\sum_{t=t^{\prime}}^{t^{\prime\prime}-1}c_{n}\left[k,t\right]a_{n}\left[k,t\right].\label{eq:thp_def}
\end{equation}
We note that $h_{n}\left[k,t\right]$, and hence $c_{n}\left[k,t\right]$ and $\upsilon\left(t^{\prime},t^{\prime\prime}\right)$, are random variables. In the sequel, their expectations will be used in the performance evaluation.

\subsection{Performance Metrics}
\subsubsection{Timeliness requirement for aerial traffic}
In this work, we use the \ac{aoi} metric\footnote{The proposed framework can be extended to alternative timeliness metrics, such as version age of information, by adapting the data-generation control.} to quantify the freshness of information received from the \ac{uav}~\cite{SalKouPap:J24}. Let $s[t]\in\{0,1\}$ denote the update-success indicator for the \ac{uav} at the end of the slot $t$, and let $G_t$ denote the generation time of the latest sample received at the receiver by time $t$. The \ac{aoi} at the control center is recursively defined as 
\begin{equation}
\tau\left[t+1\right]\triangleq\begin{cases}
t-G_t, & s[t]=1,\\
\tau\left[t\right]+1, & s[t]=0.
\end{cases}\label{eq:aoi_model}
\end{equation}
A transmission attempt is successful if the expected delivered payload accumulated since the previous success meets a quality threshold $\bar{\upsilon}$, {\em i.e.}, 
\begin{equation}
s\left[t\right]=\mathbb{I}\left\{ \mathbb{E}\left\{ \upsilon\left(t_{0},t\right)\right\} \ge\bar{\upsilon}\right\} .\label{eq:aoi_I}
\end{equation}
Here, the expectation is with respect to the channel $h_{n}[k,t^{\prime}]$ for all $n\in\mathcal{N}$, $k\in\mathcal{K}$ and $t^{\prime}\in[t_{0},\cdots,t-1]$. 

For timeliness, we impose a hard constraint on the peak information age, {\em i.e.},
\begin{equation}
\tau\left[t\right]\le\bar{\tau},\forall t\in\mathcal{T},\label{eq:aoi_def}
\end{equation}
where $\bar{\tau}\in\mathbb{Z}_+$. If the \ac{aoi} exceeds the threshold, any ongoing incomplete transmission shall be aborted.

\subsubsection{Fairness requirement for coexistence}
We regulate the aerial {\ac{rb} consumption} to ensure that each \ac{bs} has sufficient spectrum resources for terrestrial services. {The temporal \ac{rb} consumption} at \ac{bs} $n$ is {quantified by the maximum number of \acpl{rb} occupied by aerial traffic over the communication horizon}, \emph{i.e.},
\begin{equation*}
l_{n}\triangleq\max_{t\in\mathcal{T}}\sum_{k\in\mathcal{K}}a_{n}\left[k,t\right].
\end{equation*}
We define the spatiotemporal {\ac{rb} consumption} $\theta \in\mathbb{Z}_{+}$ as
\begin{equation}
\theta\triangleq\max_{n\in\mathcal{N}}l_{n}=\max_{n\in\mathcal{N},t\in\mathcal{T}}\sum_{k\in\mathcal{K}}a_{n}\left[k,t\right].\label{eq:def_stable}
\end{equation}
Equivalently, $K - \theta$ quantifies the worst-case residual spectrum available to terrestrial services.

\subsubsection{Energy efficiency for aerial traffic}
In contrast to ground \acpl{bs}, which have a stable and sufficient energy supply, the energy consumption at the \ac{uav} is tightly
constrained by its limited onboard battery capacity, making energy efficiency a critical design consideration. The
energy consumption $E \in\mathbb{R}_{+}$ is defined as
\begin{equation}
E\triangleq\sum_{n\in\mathcal{N},k\in\mathcal{K},t\in\mathcal{T}}a_{n}\left[k,t\right]p_{n}\left[k,t\right].
\label{eq:def_ee}
\end{equation}

\begin{figure*}[t]
\centering
\scalebox{1}{\input{figures/r1_paper_diagram.tex}}
\caption{{Problem transformation roadmap. Operation (i) applies the $\epsilon$-constraint method to transform the bi-objective problem $\mathscr P1$ into $\mathscr P2$, with Proposition~\ref{prop:pareto_front} identifying the admissible values that completely characterize the Pareto frontier. Operation (ii) introduces age-aware sampling control to obtain the explicit formulation $\mathscr P3$. Operation (iii) decomposes $\mathscr P3$ into the resource control problem $\mathscr P\text{3-2}$ and the timing control problem $\mathscr P\text{3-1}$. Operation (iv) applies the convex transformation and variable relaxation to obtain the convex problem $\mathscr P4$. Operation (v) reformulates its KKT conditions as scalar root-finding problems with a weighted $b$-matching structure, yielding Algorithm 1. Finally, Operation (vi) reformulates $\mathscr P\text{3-1}$ as a shortest-path problem, yielding Algorithm 2.}}
\label{fig:problem_roadmap}
\end{figure*}

\subsection{Problem Formulation}

The goal is to maintain the timeliness of aerial traffic while balancing on-board energy consumption and spectrum usage. Let $\boldsymbol{\mathcal{P}}=\{p_{n}\left[k,t\right]\}_{n\in\mathcal{N},k\in\mathcal{K},t\in\mathcal{T}}$ and $\boldsymbol{\mathcal{A}}=\{a_{n}\left[k,t\right]\}_{n\in\mathcal{N},k\in\mathcal{K},t\in\mathcal{T}}$ denote the power allocation decisions and the \ac{rb} allocation decisions over the entire planning horizon, respectively. We aim to solve the following constrained bi-objective problem\footnote{Problem $\mathscr{P}1$ serves as a canonical formulation. We note that our analysis is applicable to a broad class of variants, as detailed in Section~\ref{sec:variants}.}
\begin{align}
\mathscr{P}1:\underset{\boldsymbol{\pi}\triangleq\left[\boldsymbol{\mathcal{A}},\boldsymbol{\mathcal{P}}\right]}{\text{minimize}} & \ \left\{ \theta\left(\boldsymbol{\pi}\right),E\left(\boldsymbol{\pi}\right)\right\}, \label{eq:obj_P1}\\
\text{subject to} & \ \tau\left[t\right]\le\bar{\tau},\forall t\in\mathcal{T},\label{eq:aoi_P1_c1}\\
 & \ \boldsymbol{\pi}\in\Pi\triangleq\left(\mathcal{S}_{\mathbf{A}}\times\mathcal{S}_{\mathbf{P}}\right)^{\left|\mathcal{T}\right|},\label{eq:basic_P1_c1}
\end{align}
where \eqref{eq:obj_P1} includes the spectrum usage and on-board energy consumption objectives as defined in \eqref{eq:def_stable} and \eqref{eq:def_ee}, 
\eqref{eq:aoi_P1_c1} is the peak \ac{aoi} constraint
as defined in \eqref{eq:aoi_model}--\eqref{eq:aoi_def}, and~\eqref{eq:basic_P1_c1} is the feasible policy constraint specifying all admissible policies defined in \eqref{eq:def_fs_A}--\eqref{eq:def_fs_P}. 

 A policy is feasible if it satisfies the peak \ac{aoi} constraint. The challenge lies in finding the most effective policy that meets this constraint with minimal expenditure, thus preserving resources to optimize the objectives. This is referred to as the optimal communication timing problem, which will be made explicit in the next section.

The solution to $\mathscr{P}1$ is the complete Pareto frontier that consists of all Pareto-optimal strategies (see Definition~\ref{def:patero}). Each point on the frontier achieves a trade-off between the two competing objectives.

\begin{definition}
\label{def:patero} (Pareto optimality). A feasible strategy $\boldsymbol{\pi}^{*}$
is Pareto-optimal if there is no other feasible strategy $\boldsymbol{\pi}\neq\boldsymbol{\pi}^{*}\in\Pi$
such that $\theta\left(\boldsymbol{\pi}^{\prime}\right)\le\theta\left(\boldsymbol{\pi}\right)$
and $E\left(\boldsymbol{\pi}^{\prime}\right)\le E\left(\boldsymbol{\pi}\right)$
with at least one strict inequality. We refer to $\boldsymbol{\pi}^{*}$
as a Pareto-optimal policy in the decision space, and its objective
image $\{\theta\left(\boldsymbol{\pi}^{*}\right), E \left(\boldsymbol{\pi}^{*}\right)\}$
as a Pareto-optimal point. The collections of all such policies and points constitute the Pareto-optimal set and the Pareto frontier,
respectively.
\end{definition}

Computing the exact Pareto frontier is a non-trivial task. First, the problem is a mixed-integer program, involving binary decision variables $a_n[k, t]$ and continuous decision variables $p_n[k, t]$. Second, the objectives and constraints are non-convex, which renders most existing algorithms ineffective or inapplicable. We will address this challenge in Section~\ref{sec:decomposition} and Section~\ref{sec:algorithm}, following the roadmap illustrated in Fig.~\ref{fig:problem_roadmap}.

\section{Pareto Analysis and Problem Decomposition}\label{sec:decomposition}
This section presents our main theoretical results on the Pareto frontier. We use the $\epsilon$-constraint method ($\mathscr{P}2$) to characterize the frontier~\cite{Kai:B98} and derive tight upper and lower bounds for its region. The peak \ac{aoi} constraint in $\mathscr{P}2$ is then converted into a sum rate constraint that links age and communication timing ($\mathscr{P}3$). An important result, Proposition~\ref{prop:opt_trans}, shows that $\mathscr{P}3$ can be decomposed into a two-layer structure: the outer layer determines the optimal communication timing, while the inner layer finds the optimal resource allocation decisions for each given communication interval.

\subsection{Characterization of the Pareto Frontier}
We adopt the $\varepsilon$-constraint method to characterize the Pareto frontier. This choice is motivated by its theoretical completeness and practical implementability, particularly given the discrete nature of the system variables. In contrast, other common approaches, such as the weighted-sum method and the weighted $L_{p}$-norm method, may fail to capture non-convex regions of the frontier and can introduce additional complexity~\cite{Kai:B98}.

For a given $\varepsilon_{\theta}\in\mathcal{K}$, the $\varepsilon$-constraint method converts the bi-objective $\mathscr{P}1$ into a single-objective problem with multiple constraints. Formally, it is defined as
\begin{align*}
\mathscr{P}\text{2}:\underset{\boldsymbol{\pi}\in\Pi}{\text{min}} & \ E\left(\boldsymbol{\pi}\right),\,\,\,\text{s.t. }\tau\left[t\right]\le\bar{\tau},\forall t~\text{and }~\theta\left(\boldsymbol{\pi}\right)\le\varepsilon_{\theta}.
\end{align*}

We next show that, for any given $\varepsilon_{\theta}$ within a certain range, the solution to $\mathscr{P}\text{2}$ is Pareto-optimal. Hence, by varying $\varepsilon_{\theta}$ over its feasible range, one can characterize the Pareto frontier of $\mathscr{P}1$. Since the {spatiotemporal \ac{rb} consumption} variable $\theta$ is integer-valued and finite, its feasible range is finite, making the sweep implementable. Furthermore, we derive tight bounds on $\theta$ to reduce computational overhead. These results are summarized in the following proposition.

\begin{proposition}[Pareto frontier]\label{prop:pareto_front}
The set 
\begin{equation*}
\mathcal{C}\triangleq\left\{ \left(\varepsilon_{\theta},E^{*}\left(\varepsilon_{\theta}\right)\right):\varepsilon_{\theta}\in\left[\underline{\theta},\overline{\theta}\right]\cap \mathbb{Z}_{+}\right\},
\end{equation*}
is the Pareto frontier of $\mathscr{P}1$, where $\underline{\theta}\triangleq\theta^{\star}$,
\begin{equation}
\overline{\theta}\triangleq\min\left\{ \varepsilon_{\theta}\in\mathbb{Z}_{+}:E^{*}(\varepsilon_{\theta})=E^{\star}\right\} ,\label{eq:def_theta_up}
\end{equation}
$E^{*}\left(\varepsilon_{\theta}\right)$ is the optimal value to $\mathscr{P}\text{2}$, and $(\theta^{\star},E^{\star})$ is the ideal point of $\mathscr{P}1$, {\em i.e.}, the vector consisting of the separately attained optima of the two objectives.
\end{proposition}
\begin{IEEEproof}
{See Appendix~\ref{sec:proof_prop_pareto_front}.}
\end{IEEEproof}

\subsection{Age-Aware Sampling Control}
In this section, we convert the peak \ac{aoi} constraint into a sum rate constraint that explicitly links age and sampling timing. Let $t_{i}$ denote the sampling instant (\emph{i.e.}, data generation time) of the $i$th status packet of the \ac{uav}. The \ac{uav} starts the delivery process upon generation of a new status packet; any ongoing incomplete transmission shall be aborted. Let 
\begin{equation*}
    \mathbf{t}=\{t_{1},t_{2},\cdots,t_{I}\}
\end{equation*}
denote a sampling sequence over the entire horizon $T$, where 
\begin{equation}
    t_0=1,\,\,\, 1 \leq t_1 < t_2, \ldots < t_I \leq T,\label{eq:aoi_c_s1}
\end{equation}
and $I$ is the total number of sampling events.

The peak \ac{aoi} constraint in~\eqref{eq:aoi_def} implies that the sampling interval cannot exceed $\bar{\tau}$, that is,
\begin{equation}
1\le t_{i+1}-t_{i}\le\bar{\tau},\quad \forall i\in\mathcal{I}.\label{eq:aoi_c_s2}
\end{equation}
We set $t_{I+1}=T+1$ to represent the end time of transmission for the $I$th status update. Therefore, the feasible sampling sequence is restricted by 
\begin{equation}
\Upsilon=\Big\{ \mathbf{t}\in\mathcal{T}^{\left|\mathcal{I}\right|}:1\le t_{i+1}-t_{i}\le\bar{\tau},\forall i\in\mathcal{I}\Big\} .\label{eq:def_fs_T}
\end{equation}

Recall that the transmission is successful if the sum rate satisfies \eqref{eq:aoi_I}. Accordingly, we impose the following expected sum rate constraint between any two consecutive sampling instances 
\begin{equation}
\mathbb{E}\left\{ \upsilon\left(t_{i},t_{i+1}\right)\right\} \ge\bar{\upsilon},\quad \forall i\in\mathcal{I}.\label{eq:aoi_c_s3}
\end{equation}

As a result, the peak \ac{aoi} constraint $\tau\left[t\right]\le\bar{\tau}$ is converted into the sum rate constraint in~\eqref{eq:aoi_c_s3} with $\mathbf{t}\in\Upsilon$. Then, $\mathscr{P}\text{2}$
is equivalent to the following problem
\begin{align}
\mathscr{P}3:\underset{\boldsymbol{\mathcal{A}},\boldsymbol{\mathcal{P}},\mathbf{t}}{\text{minimize}} & \ E\left(\boldsymbol{\pi}\right),\\
\text{s.t.}~~& \ \mathbb{E}\left\{ \upsilon\left(t_{i},t_{i+1}\right)\right\} \ge\bar{\upsilon},\forall i\in\mathcal{I},\label{eq:c_thp_p3}\\
 & \ \theta\left(\boldsymbol{\pi}\right)\le\varepsilon_{\theta},\\
 & \ \boldsymbol{\mathcal{A}}\in\mathcal{S}_{\mathbf{A}}^{\left|\mathcal{T}\right|},\boldsymbol{\mathcal{P}}\in\mathcal{S}_{\mathbf{P}}^{\left|\mathcal{T}\right|},\mathbf{t}\in\Upsilon.
\end{align}

\begin{figure}
\begin{centering}
\includegraphics[width=0.85\columnwidth]{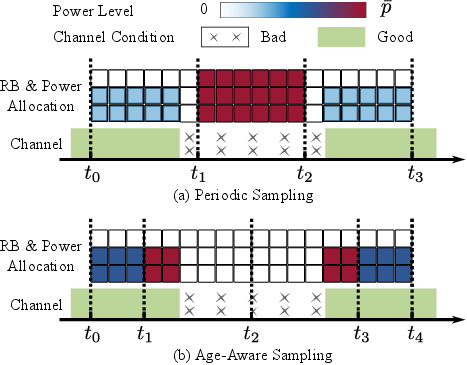}
\par\end{centering}
\caption{\label{fig:status_aware_sampling} Comparison of periodic and age-aware sampling policies. (a) The periodic policy communicates at a fixed rate and may violate the \ac{aoi} constraint under poor channel conditions. (b) The age-aware policy adjusts its sampling time and resource allocation based on predictive channel information.}
\end{figure}

Fig.~\ref{fig:status_aware_sampling} illustrates that the age-aware sampling can reduce both \acpl{rb} usage and energy consumption while maintaining the peak \ac{aoi} constraint. An important implication is that age-aware sampling may involve more transmissions while using fewer resources by exploiting predictive channel information. More comparisons are presented in Section~\ref{sec:Simulation}.

\subsection{Problem Decomposition}
Observe from problem $\mathscr{P}3$ that the variables are coupled over sampling instant $\mathbf{t}$ by the objective function and constraint \eqref{eq:c_thp_p3}. Intuitively, given any feasible sampling variable sequence $\mathbf{t}$, $\mathscr{P}3$ can be decomposed into $I$ parallel resource allocation subproblems. This result is formalized in the next theorem.

\begin{theorem}[Decomposition of $\mathscr{P}3$]\label{prop:opt_trans} 
Problem $\mathscr{P}3$ is equivalently transformed into the following outer
subproblem 
\begin{equation*}
\mathscr{P}\text{3-1}:\underset{\mathbf{t}\in\Upsilon}{\text{minimize}}\ \sum_{i\in\mathcal{I}}E^{*}\left(t_{i},t_{i+1}\right),
\end{equation*}
where $E^{*}(t_{i},t_{i+1})$ is the solution to the inner
subproblem 
\begin{align*}
\mathscr{P}\text{3-2}:\underset{\boldsymbol{\mathcal{A}}_{i},\boldsymbol{\mathcal{P}}_{i}}{\text{minimize}} & \quad\sum_{n\in\mathcal{N},k\in\mathcal{K},t\in\mathcal{T}_{i}}a_{n}\left[k,t\right]p_{n}\left[k,t\right]\\
\text{subject to} & \quad\mathbb{E}\left\{ \upsilon\left(t_{i},t_{i+1}\right)\right\} \ge\bar{\upsilon},\\
 & \quad\sum_{k\in\mathcal{K}}a_{n}\left[k,t\right]\le\varepsilon_{\theta},\forall n\in\mathcal{N},t\in\mathcal{T}_{i},\\
 & \quad\boldsymbol{\mathcal{A}}_{i}\in\mathcal{S}_{\mathbf{A}}^{\left|\mathcal{T}_{i}\right|},\boldsymbol{\mathcal{P}}_{i}\in\mathcal{S}_{\mathbf{P}}^{\left|\mathcal{T}_{i}\right|},
\end{align*}
{with $E^*(t_i,t_{i+1})=+\infty$ if the interval is infeasible,} where $\mathcal{T}_{i} = [t_{i},t_{i+1})\cap\mathbb{Z}_{+}$
denotes the $i$th communication interval, and $\boldsymbol{\mathcal{A}}_{i}=\{a_{n}\left[k,t\right]\}_{n\in\mathcal{N},k\in\mathcal{K},t\in\mathcal{T}_{i}}$
and $\boldsymbol{\mathcal{P}}_{i}=\{p_{n}\left[k,t\right]\}_{n\in\mathcal{N},k\in\mathcal{K},t\in\mathcal{T}_{i}}$
are the \ac{rb} and power allocation policies restricted to the time interval $\mathcal{T}_{i}$, respectively.
\end{theorem}
\begin{IEEEproof}
{See Appendix~\ref{sec:proof_prop_opt_trans}.}
\end{IEEEproof}

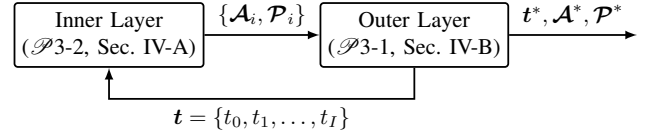
\begin{figure}[t!]
    \centering
    \scalebox{0.85}{\input{figures/two_layer}}
    \caption{Schematic representation of the two-layer structure.}
    \label{fig:two_layer}
\end{figure}

As depicted in Fig.~\ref{fig:two_layer}, instead of directly solving the high-dimensional Problem $\mathscr{P}3$, we propose a two-layer structure that isolates the resource allocation problem ($\mathscr{P}\text{3-2}$) from the sampling control problem ($\mathscr{P}\text{3-1}$). Note that a sampling sequence $\mathbf{t}$ is feasible if and only if all of its inner problems are feasible. {In short, the inner layer optimizes resource allocation for each candidate sampling interval and uses its cost as the edge weight in the outer-layer graph. The outer layer finds the shortest path of the constructed graph to determine the sampling schedule and retrieves the corresponding resource-allocation policies.} The next section develops efficient algorithms to solve these sub-problems.

\section{Graph-Based Low-Complexity Algorithm}\label{sec:algorithm}
This section develops efficient algorithms to solve the decomposed subproblems $\mathscr{P}\text{3-1}$ and $\mathscr{P}\text{3-2}$. We first give a brief overview of the proposed approach.

Section~\ref{subsec:physical_layer_control} focuses on the inner problem of power and access control for a given communication interval. The main difficulty lies in the fact that $\mathscr{P}\text{3-2}$ is a non-convex, mixed-integer program. To address this issue, we first relax the problem into a continuous-variable convex formulation $\mathscr{P}\text{4}$, which exhibits several convenient properties that can be leveraged to construct an optimal policy for $\mathscr{P}\text{3-2}$, and develop an efficient algorithm by lower-bounding the expected capacity.

Section~\ref{subsec:graph} addresses the outer problem of sampling timing control. We show that $\mathscr{P}\text{3-1}$ is equivalent to a shortest-path problem on a finite graph, where the weight of each edge corresponds to the solution of the inner problem for a given communication interval. A graph-based algorithm, Algorithm~\ref{alg:g_control_alg}, is proposed to find an optimal sampling sequence.

{Section~\ref{subsec:alg_analysis} establishes asymptotic optimality and shows that the computational complexity scales linearly with the mission horizon.}

\subsection{Inner Solution: Power and Access Control}\label{subsec:physical_layer_control}

\subsubsection{Convex Relaxation and Transformation\label{subsec:convex_r_t}}
Firstly, we relax the binary \ac{rb} allocation constraint such that
$a_{n}\left[k,t\right]\in\left[0,1\right]$.
Then, to address the non-convexity inherent in the objective and rate-constrained
resource allocation (specifically {$c_n[k,t]a_n[k,t]$ in
\eqref{eq:c_thp_p3}}), we employ a variable transformation strategy
inspired by \cite{LiChe:J24a}. We introduce an auxiliary variable {$\phi_{n}\left[k,t\right]=\bar{c}_n[k,t]a_{n}\left[k,t\right]$, where $\bar{c}_n[k,t]\triangleq\mathbb{E}\{{c}_{n}\left[k,t\right]\}$.
Denote the inverse function of $\bar{c}_n[k,t]$ over $p_n[k,t]$ as $c_{n,k,t}^{-1}(p_n[k,t])$,}
then, the transmission power variable is transformed to 
\begin{equation}
p_{n}\left[k,t\right]=\begin{cases}
c_{n,k,t}^{-1}\left(\frac{\phi_{n}\left[k,t\right]}{a_{n}\left[k,t\right]}\right), & a_{n}\left[k,t\right]>0,\\
0, & a_{n}\left[k,t\right]=0.
\end{cases}\label{eq:def_p_over_a_phi}
\end{equation}

Substituting \eqref{eq:def_p_over_a_phi} into the energy consumption formula in \eqref{eq:def_ee}, the
objective function of $\mathscr{P}\text{3-2}$ is reformulated as
\begin{equation*}
\sum_{n\in\mathcal{N},k\in\mathcal{K},t\in\mathcal{T}_{i}}
a_{n}\left[k,t\right]c_{n,k,t}^{-1}\left(\frac{\phi_{n}\left[k,t\right]}{a_{n}\left[k,t\right]}\right),
\end{equation*}
which is convex, as each summand is the perspective function
of the convex function $c_{n,k,t}^{-1}\left(p_n[k,t]\right)$
\cite{Boy:B04}. 
Indeed, $\bar{c}_{n,k,t}^\prime\triangleq\frac{\partial\bar{c}_n[k,t]}{\partial p_n[k,t]}>0$ and $\bar{c}_{n,k,t}^{\prime\prime}\triangleq\frac{\partial^2\bar{c}_n[k,t]}{\partial p^2_n[k,t]}<0$, yielding 
$
\frac{\partial^2c_{n,k,t}^{-1}}{\partial p^2_{n}[k,t]} = -\frac{\bar{c}_{n,k,t}^{\prime\prime}}{(\bar{c}_{n,k,t}^\prime)^3}>0.
$

As a result, the original sum-rate constraint \eqref{eq:c_thp_p3} becomes linear 
\begin{equation*}
\sum_{n\in\mathcal{N},k\in\mathcal{K},t\in\mathcal{T}_{i}}\phi_{n}\left[k,t\right]\ge\bar{\upsilon}.
\end{equation*}
Similarly, the power constraints $\sum_{n,k}p_{n}\left[k,t\right]\le\bar{p}$ and $p_{n}\left[k,t\right]\ge0$ are converted to
\begin{equation*}
\sum_{n\in\mathcal{N},k\in\mathcal{K}}
a_{n}\left[k,t\right]c_{n,k,t}^{-1}\left(\frac{\phi_{n}\left[k,t\right]}{a_{n}\left[k,t\right]}\right)
\le\bar{p},
\,\,\, \forall t\in\mathcal{T}_{i},
\end{equation*}
and $\phi_{n}\left[k,t\right]\ge0$, respectively.

By combining the
continuous relaxation and variable transformation, $\mathscr{P}\text{3-2}$ is reformulated such that the objective function and all constraints become convex. The resulting relaxed problem is therefore convex, and its formulation is given as follows.
\begin{align}
\mathscr{P}4:&~ \underset{\boldsymbol{\mathcal{A}}_{i},\boldsymbol{\Phi}_{i}}{\text{minimize}} 
 \, \sum_{n\in\mathcal{N},k\in\mathcal{K},t\in\mathcal{T}_{i}}
a_{n}\left[k,t\right]c_{n,k,t}^{-1}\left(\frac{\phi_{n}\left[k,t\right]}{a_{n}\left[k,t\right]}\right)\nonumber \\
\text{s.t.}~& \sum_{n\in\mathcal{N},k\in\mathcal{K},t\in\mathcal{T}_{i}}\phi_{n}\left[k,t\right]\ge\bar{\upsilon}, \\
&\phi_{n}\left[k,t\right]\ge0,\forall n\in\mathcal{N},k\in\mathcal{K},t\in\mathcal{T}_{i}, \\
&\sum_{n\in\mathcal{N},k\in\mathcal{K}}
a_{n}\left[k,t\right]c_{n,k,t}^{-1}\left(\frac{\phi_{n}\left[k,t\right]}{a_{n}\left[k,t\right]}\right)
\le\bar{p},\forall t\in\mathcal{T}_{i},\\
 & \sum_{k\in\mathcal{K}}a_{n}\left[k,t\right]\le\varepsilon_{\theta},\forall n\in\mathcal{N},t\in\mathcal{T}_{i},\label{eq:cp_c_a1}\\
 & \sum_{n\in\mathcal{N}}a_{n}\left[k,t\right]\le1,\forall k\in\mathcal{K},t\in\mathcal{T}_{i},\label{eq:cp_c_a2}\\
 & a_{n}\left[k,t\right]\in\left[0,1\right],\forall n\in\mathcal{N},k\in\mathcal{K},t\in\mathcal{T}_{i},\label{eq:cp_c_a3}
\end{align}
where $\boldsymbol{\Phi}_{i}\triangleq\{\phi_{n}\left[k,t\right]\}_{n\in\mathcal{N},k\in\mathcal{K},t\in\mathcal{T}_{i}}$.

\subsubsection{Efficient and Binary-Enforcing Algorithm} \label{subsubsec:efficient_alg}
{Although the transformed problem is convex, the inverse expected-capacity function lacks a closed-form expression, precluding a direct solution using standard commercial solvers.}
More importantly, they produce non-deterministic access decisions because the resulting variables $a_{n}[k,t]$ are continuous rather than binary. We are therefore motivated to develop a specialized, efficient algorithm that enforces binary constraints. The main idea is to analyze the \ac{kkt} conditions of the relaxed problem $\mathscr{P}4$ and exploit the resulting structure to retrieve the optimal binary solution.

The next proposition derives the optima for $\mathscr{P}4$ based on \ac{kkt} analysis. It shows that the power allocation follows a generalized capped water-filling rule, and the \ac{rb} assignment becomes a water-level-governed linear program.

\begin{proposition}[Optimal resource allocation]\label{prop:opt_p_a}
There exists an optimal
solution to $\mathscr{P}4$ such that for each $t\in\mathcal{T}_{i}$,
the optimal \ac{rb} allocation $\mathbf{A}\left[t;\lambda_{t}\right]\triangleq\{a_{n}[k,t;\lambda_{t}]\}_{n,k}$
belongs to the optimal solution set of the per-slot linear program
\begin{equation}
\mathcal{A}^{*}\left(t;\lambda_{t}\right)\triangleq\arg\underset{\mathbf{A}\left[t\right]\in\mathcal{A}\left(t\right)}{\min}\sum_{n\in\mathcal{N},k\in\mathcal{K}}w_{n}\left[k,t;\lambda_{t}\right]a_{n}\left[k,t\right],\label{eq:opt_a}
\end{equation}
with weights 
\begin{equation*}
w_{n}\left[k,t;\lambda_{t}\right]=p_{n}\left[k,t;\lambda_{t}\right]-\ln2\lambda_{t}\bar{c}_{n}\left[k,t;\lambda_{t}\right],
\end{equation*}
where 
\begin{equation*}
\mathcal{A}\left(t\right)\triangleq\left\{ \left\{ a_{n}\left[k,t\right]\right\} _{n,k}:\text{\eqref{eq:cp_c_a1}--\eqref{eq:cp_c_a3}}\right\}.
\end{equation*}
Here, the optimal power and capacity follow
\begin{equation}
p_{n}\left[k,t;\lambda_{t}\right]=\left[\tilde{p}_n[k,t;\lambda_t]\right]^{+}, \label{eq:opt_p_general}
\end{equation}
\begin{equation}
    \bar{c}_{n}[k,t;\lambda_t]
=\mathbb{E}\Big\{\log_2\Big(1+\frac{p_{n}\left[k,t;\lambda_{t}\right]g_{n}\left[k,t\right]\zeta_{n}\left[k,t\right]}{\delta^2}\Big)\Big\}, \label{eq:opt_c_general}
\end{equation}
where $[x]^{+}\triangleq\max(0,x)$, $\tilde{p}_n[k,t;\lambda_t]$ is the solution to
\begin{equation}
    \mathbb{E}\left\{\frac{g_n[k,t]\zeta_n[k,t]}{\delta^2+p_n[k,t]g_n[k,t]\zeta_n[k,t]}\right\}= \lambda_t^{-1}.
    \label{eq:opt_p_general_equation}
\end{equation}
and the effective per-slot water level is capped as
\begin{equation*}
\lambda_{t}=\min\left\{ \frac{\lambda^{*}}{\ln2},\tilde{\lambda}_{t}\right\} .
\end{equation*}
The per-slot caping level $\tilde{\lambda}_{t}$ is chosen to satisfy
\begin{align}
P\left(t;\lambda\right) & \triangleq\sum_{n\in\mathcal{N},k\in\mathcal{K}}a_{n}\left[k,t;\lambda\right]p_{n}\left[k,t;\lambda\right]=\bar{p}.\label{eq:equation_for_clip}
\end{align}
The global water level \textup{$\lambda^{*}$} is determined by
\begin{align}
\Phi\left(\lambda\right) & \triangleq\sum_{t\in\mathcal{T}_{i}}\Phi\left(t;\lambda\right)\nonumber \\
 & \triangleq\sum_{t\in\mathcal{T}_{i}}\sum_{n\in\mathcal{N},k\in\mathcal{K}}a_{n}\left[k,t;\lambda_{t}\right]\bar{c}_{n}[k,t;\lambda_t]=\bar{\upsilon}.\label{eq:equation_for_water_level}
\end{align}
\end{proposition}
\begin{IEEEproof}
{See Appendix~\ref{sec:proof_prop_opt_p_a}.}
\end{IEEEproof}

As a result, Problem $\mathscr{P}4$ is reduced to solving equations \eqref{eq:equation_for_clip}
and \eqref{eq:equation_for_water_level}. To address the potential \emph{non-uniqueness} in $\mathcal{A}^{*}\left(t;\lambda_{t}\right)$, which may yield multi-valued and discontinuous functions $P\left(t;\lambda\right)$ and $\Phi\left(\lambda\right)$, we define the extended functions $P_{\text{ext}}\left(t;\lambda,\xi\right)$ and $\Phi_{\text{ext}}\left(\lambda,\xi\right)$
using a convex combination of extreme points. Proposition~\ref{prop:property_cc} shows these extensions are optimal, continuous, and non-decreasing in $\lambda$ and $\xi$.
Additionally, the existence of a binary \ac{rb} allocation solution for any $\lambda$ is proven in Proposition~\ref{prop:Integrality}, which enables the construction of a feasible binary policy.

\begin{proposition}[Binary \ac{rb} allocation]\label{prop:Integrality}
For any $\lambda$, there exists a solution $\mathbf{A}\left[t;\lambda_{t}\right]\in\mathcal{A}^{*}\left(t;\lambda\right)$ such that $a_{n}[k,t;\lambda]\in\left\{ 0,1\right\} $ for all $n\in\mathcal{N},k\in\mathcal{K},t\in\mathcal{T}_{i}$.
\end{proposition}
\begin{IEEEproof}
{See Appendix~\ref{sec:proof_prop_opt_p_a}.}
\end{IEEEproof}

As a result, Problem \eqref{eq:opt_a} can
be solved by the b-matching algorithm. While Proposition \ref{prop:Integrality} guarantees the existence of binary optimal solutions, the set $\mathcal{A}^{*}\left(t;\lambda_{t}^{*}\right)$
may contain multiple optimal \ac{rb} allocations at certain critical water levels.

Let $\lambda_{0}$ denote such a \emph{critical point}. We define two limiting binary solutions $\mathbf{A}^{-}\left[t;\lambda_{0}\right]\triangleq\lim_{\lambda\to\lambda_{0}^-}\mathbf{A}\left[t;\lambda\right]$
and $\mathbf{A}^{+}\left[t;\lambda_{0}\right]\triangleq\lim_{\lambda\to\lambda_{0}^+}\mathbf{A}\left[t;\lambda\right]$,
which always exist according to Proposition~\ref{prop:Integrality}. For any $\xi\in\left[0,1\right]$, we construct the convexified \ac{rb} allocation as $\mathbf{A}\left[t;\lambda_{0},\xi\right]=(1-\xi)\mathbf{A}^{-}\left[t;\lambda_{0}\right]+\xi\mathbf{A}^{+}\left[t;\lambda_{0}\right]$. Accordingly, the extended per-slot power and sum rate are defined, respectively, as
\begin{equation*}
P_{\text{ext}}\left(t;\lambda,\xi\right)=\sum_{n\in\mathcal{N},k\in\mathcal{K}}a_{n}\left[k,t;\lambda,\xi\right]p_{n}\left[k,t;\lambda_{t}\right],
\end{equation*}
\begin{equation*}
\Phi_{\text{ext}}\left(t;\lambda,\xi\right)=\sum_{n\in\mathcal{N},k\in\mathcal{K}}a_{n}\left[k,t;\lambda,\xi\right]\bar{c}_{n}\left[k,t;\lambda_{t}\right].
\end{equation*}

\begin{proposition}[Properties of extended functions]\label{prop:property_cc}
Define the extended domain $\mathcal{D}\triangleq\left\{ \left(\lambda,\xi\right):\lambda\in\mathbb{R}_{0+},\xi\in\left[0,1\right]\right\} $.
Then, the following properties hold for $P_{\text{ext}}\left(t;\lambda,\xi\right)$ (similar conditions hold for $\Phi_{\text{ext}}\left(t;\lambda,\xi\right)$).

1. Optimality: $\mathbf{A}\left[t;\lambda,\xi\right]\in\mathcal{A}^{*}\left(t;\lambda\right)$
for all $\xi\in\left[0,1\right]$; and $P_{\text{ext}}\left(t;\lambda,\xi\right)=\left(1-\xi\right)P\left(t;\lambda^{-}\right)+\xi P\left(t;\lambda^{+}\right)$.

2. Continuity: The extended functions are continuous with respect
to $\xi\in\left[0,1\right]$. Furthermore, at any point $\lambda_{0}$
where discontinuity occurs in $\lambda$, the boundaries satisfy 
\begin{equation*}
\lim_{\lambda\to\lambda_{0}^-}P_{\text{ext}}\left(t;\lambda,\xi\right)=P_{\text{ext}}\left(t;\lambda_{0},0\right),
\end{equation*}
\begin{equation*}
\lim_{\lambda\to\lambda_{0}^+}P_{\text{ext}}\left(t;\lambda,\xi\right)=P_{\text{ext}}\left(t;\lambda_{0},1\right).
\end{equation*}

\begin{algorithm}[t!]
\# Input: $\bar{\upsilon}$, $\varepsilon_{\theta}$, $g_{n}[k,t]$
and $\kappa_{n}[k,t]$ for all $n, k$ and $t \in \mathcal{T}_i$
\begin{enumerate}
\item Obtain $\tilde{\lambda}_{t}$ and $\tilde{\xi}_{t}$ by computing
$P_{\text{ext}}\left(t;\lambda,\xi\right)=\bar{p}$ via bi-section
search for all $t$.
\item Obtain $\lambda^{-}$ and $\lambda^{+}$ with $\Phi_{\text{ext}}\left(\lambda_{t}^{-},0\right)\le\bar{v}\le\Phi_{\text{ext}}\left(\lambda_{t}^{+},0\right)$
via bi-section search.
\item Obtain $\xi^{*}$ by computing $\Phi_{\text{ext}}\left(\lambda_{t},\xi_{t}\right)=\bar{v}$
via bi-section search, with $\xi_{t}=\min\{\xi,\tilde{\xi}_{t}\}$.
\end{enumerate}
\# Output: $p_{n}^{*}\left[k,t\right] \gets p_{n}\left[k,t;\lambda_{t}\right]$,
$a_{n}^{*}\left[k,t\right] \gets a_{n}\left[k,t;\lambda_{t}^{-},\xi^{*}\right]$,
and $E^{*}\left(t_{i},t_{i+1}\right) \gets \sum_{t\in\mathcal{T}_{i}}P_{\text{ext}}\left(t;\lambda_{t}^{-},\xi_{t}^{*}\right).$
\caption{\label{alg:opt_power_control}
Efficient power and access control.}
\end{algorithm}

\begin{algorithm}[t]
\# Input: $\theta$, $g_{n}[k,t]$ and $\kappa_{n}[k,t]$ for all $n, k, t$
\begin{enumerate}
\item Construct the graph $\mathscr{G}$ based on Fig.~\ref{fig:Sample-timing-graph}. Compute the weights by solving $\mathscr{P}\text{3-2}$ using Algorithm~\ref{alg:opt_power_control}.
\item Find the optimal sampling time sequence $\mathbf{t}^{*}$ using the shortest-path algorithm on the graph $\mathscr{G}$.
\end{enumerate}
\# Output: $\boldsymbol{\mathcal{A}}^*$, $\boldsymbol{\mathcal{P}}^*$, $\mathbf{t}^{*}$, and
$E^{*}(\theta)$
\caption{\label{alg:g_control_alg}
Graph-based sampling timing control.}
\end{algorithm}

3. Monotonicity: $P_{\text{ext}}\left(t;\lambda_{1},\xi_{1}\right)\le P_{\text{ext}}\left(t;\lambda_{2},\xi_{2}\right)$
if $\lambda_{1}\le\lambda_{2}$ for any $\xi_{1}$ and $\xi_{2}$;
and $P_{\text{ext}}\left(t;\lambda,\xi_{1}\right)\le P_{\text{ext}}\left(t;\lambda,\xi_{2}\right)$
if $\xi_{1}\le\xi_{2}$ for any $\lambda$.
\end{proposition}
\begin{IEEEproof}
{See Appendix~\ref{sec:proof_prop_opt_p_a}.}
\end{IEEEproof}

Building on the results in Propositions
\ref{prop:opt_p_a}--\ref{prop:property_cc}, we propose a power and access control algorithm (Algorithm~\ref{alg:opt_power_control}), which determines the optimal solution to the relaxed Problem $\mathscr{P}4$
by effectively solving the equations $P_{\text{ext}}(t;\lambda,\xi)=\bar{p}$
for all $t$ and $\Phi_{\text{ext}}(\lambda_{t},\xi_{t})\triangleq\sum_{t}\Phi_{\text{ext}}\left(t;\lambda_{t},\xi_{t}\right)=\bar{v}$.
The procedure executes a two-stage iterative search: it first computes
the local caping thresholds $\tilde{\lambda}_{t}$ and combining parameters $\tilde{\xi}_{t}$, followed by the determination of the global water-level
$\lambda^{*}$ and instantaneous $\xi_{t}^{*}$. For ensuring the binary of \ac{rb} allocation, one can simply choose
the policy based on $\lambda_{t}^{+}$ and the energy consumption is $\sum_{t\in\mathcal{T}_{i}}P_{\text{ext}}\left(t;\lambda_{t}^{+},0\right)$.

{Proposition~\ref{prop:property_cc} establishes the convergence and optimality of Algorithm~\ref{alg:opt_power_control}. Specifically, continuity guarantees the existence of solutions to $P_{\mathrm{ext}}(t;\lambda,\xi)=\bar{p}$ and $\Phi_{\mathrm{ext}}(\lambda_t,\xi_t)=\bar{v}$, monotonicity ensures convergence of the bisection searches, and the optimality condition establishes that the resulting solution solves $\mathscr{P}4$ optimally.}

\subsubsection{A Tractable Lower Bound of the Expected Sum Rate} \label{subsubsec:low-complexity_alg}
{Evaluating \eqref{eq:opt_p_general} and \eqref{eq:opt_c_general} requires repeated numerical integration and root finding, which is computationally intensive. We therefore derive a tractable deterministic lower bound on the expected sum rate, yielding closed-form counterparts to \eqref{eq:opt_p_relax} and \eqref{eq:opt_c_relax} and an efficient robust algorithm.}

\begin{lemma}[Explicit lower bound]\label{lem:c_lb}
The expected capacity
$\mathbb{E}\{c_{n}\left[k,t\right]\}$ is lower bounded by 
\begin{align}
\mathbb{E}\left\{ c_{n}\left[k,t\right]\right\} \ge\underbar{c}_{n}\left[k,t\right] \triangleq\log_{2}\left(1+\beta_{n}\left[k,t\right]\bar{\gamma}_{n}\left[k,t\right]\right),
\label{eq:def_c_bar}
\end{align}
where $\bar{\gamma}_{n}\left[k,t\right]=p_{n}\left[k,t\right]g_{n}\left[k,t\right]/\delta^{2}$
denotes the average received \ac{snr}, and {$\beta_{n}\left[k,t\right]=e^{\psi\left(\kappa_{n}\left[k,t\right]\right)-\sigma^2_n[k,t]/2}/\kappa_{n}\left[k,t\right]$} characterizes the capacity loss due to {uncertainty}, where $\psi(\cdot)$ is the Digamma function. Moreover, the bound is tight in high-\ac{snr} or {low-uncertainty} regime, i.e., $\mathbb{E}\left\{ c_{n}\left[k,t\right]\right\} \to \underbar{c}_{n}\left[k,t\right]$ as $\bar{\gamma}_{n}\left[k,t\right] \to \infty$ or $\kappa_{n}\left[k,t\right] \to \infty$ {and $\sigma^2_n[k,t]\to0$}.
\end{lemma}
\begin{IEEEproof}
See Appendix~\ref{sec:proof_lem_c_lb}.
\end{IEEEproof}

\begin{corollary}[Robust power allocation]\label{cor:robust_p}
By Lemma~\ref{lem:c_lb}, the sum rate constraint can be relaxed as
\begin{equation*}
\mathbb{E}\{v(t_i),t_{i+1}\} \ge  \sum_{n\in\mathcal{N},k\in\mathcal{K},t\in\mathcal{T}_{i}}\underbar{c}_{n}\left[k,t\right]a_{n}\left[k,t\right] \ge\bar{\upsilon},
\end{equation*}
and the corresponding solution has the same structure as that in Proposition~\ref{prop:opt_p_a}, except that \eqref{eq:opt_p_general} and \eqref{eq:opt_c_general} are replaced by
\begin{equation}
p_{n}\left[k,t;\lambda_{t}\right]=\left[\lambda_{t}-\iota_{n}\left[k,t\right]\right]^{+},
\label{eq:opt_p_relax}
\end{equation}
\begin{equation}
\underbar{c}_{n}\left[k,t;\lambda_{t}\right]=\left[\log_{2}\left(\lambda_{t}\right)-\log_{2}\left(\iota_{n}\left[k,t\right]\right)\right]^{+},
\label{eq:opt_c_relax}
\end{equation}
where $\iota_{n}\left[k,t\right]=\delta^{2}/(\beta_{n}\left[k,t\right]g_{n}\left[k,t\right])$
is a constant.
\end{corollary}
\begin{IEEEproof}
The result follows by replacing $\bar{c}_n[k,t]$ with $\underbar{c}_n[k,t]$ in Section~\ref{subsec:convex_r_t} and \ref{subsubsec:efficient_alg}.
\end{IEEEproof}

Lemma~\ref{lem:c_lb} and Corollary~\ref{cor:robust_p} indicate that greater prediction uncertainty degrades effective channel quality and increases the allocation cost of affected \acpl{bs}, shifting resources toward \acpl{bs} and transmission intervals with higher average \ac{snr} and lower uncertainty.


\subsection{Outer Solution: Sampling Timing Control}
\label{subsec:graph}
This section addresses $\mathscr{P}\text{3-1}$ by transforming it into a shortest path problem on a weighted directed acyclic graph.

\subsubsection{Graph Construction}
We first construct a timing-control graph $\mathscr{G}=\{\mathbf{v},\mathbf{e},\mathbf{w}\}$,
as depicted in Fig.~\ref{fig:Sample-timing-graph}. The vertex set $\mathbf{v}=\{1,\cdots, T+1\}$ represents the discrete time indices of all possible sampling instants. The terminal node $T+1$ is introduced to mark the boundary of the optimization horizon, ensuring the final transmission interval is accounted for. The directed edge set $\mathbf{e}=\{(v_{i},v_{j})\}$
represents valid transmission intervals. An edge exists from node $i$ to node $j$ (where $v_{i},v_{j}\in\mathbf{v}$) if and only if the interval duration satisfies the \ac{aoi} constraint, {\em i.e.}, $1\le t_{i}-t_{j}\le\bar{\tau}$.
The weight set $\mathbf{w}=\{w_{i,j}\}$ assigns a cost to each edge, defined as the minimum energy consumption required for complete data delivery over the interval $[t_{i},t_{i+1})\cap\mathbb{Z}_{+}$. Specifically, $w_{i,j}=E^{*}(t_{i},t_{j})$, which is obtained by solving the inner problem $\mathscr{P}\text{3-2}$.

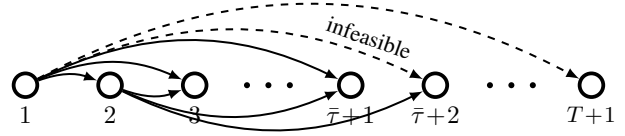
\begin{figure}[t!]
    \centering
    \scalebox{0.95}{\input{figures/graph_construction}}
    \caption{Construction of the timing-control graph. Each vertex represents a sampling instant, and each directed edge denotes a communication interval. The weight of each feasible edge is obtained by solving the inner optimization problem associated with its endpoints.}
    \label{fig:Sample-timing-graph}
\end{figure}

\subsubsection{Graph Equivalence}
By construction, any feasible sampling sequence $\mathbf{t}=\{1=t_{0}<t_{1}<\cdots<t_{I+1}=T+1\}$ maps bijectively to a directed path $1\to t_{1}\to\cdots\to t_{I+1}=T+1$ in $\mathscr{G}$. Each edge $(t_{i},t_{j})$ in $\mathscr{G}$ satisfies the \ac{aoi} constraint $1\le t_{j}-t_{i}\le\bar{\tau}$, as shown in Fig.~\ref{fig:Sample-timing-graph}. Thus,
any path from node $1$ to node $T+1$ represents a sampling sequence that satisfies $t_{0}=1$, $t_{I+1}=T+1$, and $t_{k}\in\mathcal{T}, \forall k$. Conversely, any feasible solution
in $\Upsilon$ constitutes a valid path from node $0$ to node $T+1$ in $\mathscr{G}$, and vice versa.

The weight of each feasible edge $(t_{i},t_{i+1})$ is defined as $w_{i,i+1}=E^{*}(t_{i},t_{i+1})$. The total length of a path is $\sum_{i}w_{i,i+1}$, which
equals the objective function $\sum_{i}E^{*}(t_{i},t_{i+1})$. Hence, $\mathscr{P}\text{3-1}$ is equivalent to a graph shortest-path problem. This equivalence is formalized in the following proposition.

\begin{proposition}[Equivalence of $\mathscr{P}\text{3-1}$]\label{prop:opt_P21}
Finding the optimal sampling sequence $\mathbf{t}^{*}$ for $\mathscr{P}\text{3-1}$
is equivalent to finding the shortest path from node $1$ to node
$T+1$ in $\mathscr{G}$.
\end{proposition}

\subsubsection{Graph-Based Algorithm}
As a result, Problem $\mathscr{P}\text{3-1}$ is equivalent to a shortest path problem on a weighted directed graph. The graph-based algorithm is summarized in Algorithm~\ref{alg:g_control_alg}. It proceeds in two phases: (i) construct the timing-control graph $\mathscr{G}$ with edge weights given by the optimal interval energy $E^*(\cdot,\cdot)$; and (ii) find the shortest path from node $1$ to node $T+1$. 

\subsection{Optimality and Complexity Analysis}\label{subsec:alg_analysis}

\subsubsection{Optimality and feasibility guarantees}
All transformations in MPComm are exact except $a_n[k,t]\in\{0,1\}\rightarrow a_n[k,t]\in[0,1]$. Algorithm~$\ref{alg:opt_power_control}$ globally solves the resulting convex problem $\mathscr P4$. When $\xi=0$, its solution is deterministic and globally optimal for the original mixed-integer inner problem; otherwise, the recovery yields a feasible deterministic solution. Since the two-layer decomposition and shortest-path reformulation are exact, Algorithm 2 inherits these guarantees: it globally solves $\mathscr P1$ if $\xi=0$ for every selected interval and otherwise returns a feasible solution.

The low-complexity algorithm further replaces the expected capacity with the lower bound in Lemma~\ref{lem:c_lb}, which Corollary~\ref{cor:robust_p} proves asymptotically tight in the high-\ac{snr} or low-uncertainty regime. Hence, when $\xi=0$, this algorithm is asymptotically optimal for $\mathscr P1$; otherwise, its recovered deterministic policy remains feasible.

\subsubsection{Computational Complexity}
Let $X$ denote the per-link complexity of evaluating the expected capacity \eqref{eq:opt_c_general} and solving the associated scalar equation \eqref{eq:opt_p_general_equation}. For each trial water level, evaluating all $NK$ links and solving the weighted bipartite matching~\cite{AhuMagOrl:B93} require
\begin{equation*}
C_{\mathrm{slot}} = \mathcal O \left( XNK + NK^2+(NK+K^2)\log(N+K) \right).
\end{equation*} 
Since bisection with accuracy $\epsilon$ requires $\mathcal O(\log(\epsilon^{-1}))$ iterations, Algorithm~$\ref{alg:opt_power_control}$ over $\mathcal T_i$ has complexity $
C_{\mathrm{Alg.\,1}} = \mathcal O\left( |\mathcal T_i|\,C_{\mathrm{slot}} \log{\epsilon^{-1}} \right).$

Algorithm 2 constructs the timing-control graph and finds its shortest path. For an interval length $c\in\{1,\ldots,\bar{\tau}\}$, the graph contains $T+1-c$ corresponding edges, and computing each edge weight requires $ \mathcal O \left(cC_{\mathrm{slot}} \log{\epsilon^{-1}}\right)$.
Therefore, the total edge-weight computation has complexity
\begin{equation*} 
    \sum_{c=1}^{\bar{\tau}} (T+1-c)\, \mathcal O \left(cC_{\mathrm{slot}}\log{\epsilon^{-1}}\right)
    = \mathcal O \left( T\bar{\tau}^{2}C_{\mathrm{slot}}\log{\epsilon^{-1}} \right).
\end{equation*}
The graph has $|\mathbf v|=T+1$ vertices and $|\mathbf e| = \sum_{c=1}^{\bar{\tau}}(T+1-c) = \mathcal O(T\bar{\tau})
$ edges. Since it is a directed acyclic graph, its shortest path can be found in $ \mathcal O(|\mathbf v|+|\mathbf e|) = \mathcal O(T\bar{\tau}),$
which is dominated by edge-weight computation. Hence, the overall complexity of Algorithm 2 is $C_{\mathrm{Alg.\,2}} = \mathcal O \left( T\bar{\tau}^{2} C_{\mathrm{slot}}\log{\epsilon^{-1}} \right).$

The low-complexity algorithm replaces the numerical evaluation of \eqref{eq:opt_c_general} and root-finding for \eqref{eq:opt_p_general_equation} with the closed-form expressions \eqref{eq:opt_p_relax} and \eqref{eq:opt_c_relax}, giving  $X=1 $. Hence, computing the complete Pareto frontier over  $K $ values of  $\theta $ requires
\begin{equation*} 
\mathcal O \left( T\bar{\tau}^{2}\log{\epsilon}^{-1} \left[ NK^{3}+(NK^{2}+K^{3})\log(N+K) \right] \right).
\end{equation*}
\color{black}

\section{Extension to Problem Variants}\label{sec:variants}
Problem $\mathscr{P}1$ is adopted as a canonical formulation. In practice, one may encounter a broad class of variants, such that objective transformations and single-objective specializations. We next discuss three representative variants and specify conditions under which our approach remains applicable.

\subsection{Objective Transformation}
To represent heterogeneous requirements, we introduce a generic objective vector 
\begin{equation*}
\mathbf{g}\left(\theta,E\right)=\left(g_{1}\left(\theta\right),g_{2}\left(E\right)\right)^{\text{T}}:\mathbb{R}^{2}\to\mathbb{R}^{2},
\end{equation*}
where $g_{1}(\theta):\mathbb{R}\to\mathbb{R}$ penalizes the spatiotemporal $\theta$, and $g_{2}(E):\mathbb{R}\to\mathbb{R}$ penalizes
the \ac{uav} energy consumption $E$. The resulting nonlinear transformed problem is 
\begin{equation*}
\mathscr{P}_1^\mathrm{v}:\min_{\boldsymbol{\pi}\in\Pi}\ \left\{ g_{1}\left(\theta\left(\boldsymbol{\pi}\right)\right),g_{2}\left(E\left(\boldsymbol{\pi}\right)\right)\right\} ,\text{s.t. }\tau\left[t\right]\le\bar{\tau},\forall t\in\mathcal{T}.
\end{equation*}

Some practical choices of $\mathbf{g}$ include (i) linear models for constant marginal costs \cite{GatRibPap:J14}, (i) barrier-type metrics for delay-sensitive congestion \cite{Ros:B14}, and (iii) non-convex logistic forms representing saturation or risk thresholds
~\cite{LeeKwo:J09,VaiSStaKou:C25}. 

Let $\mathcal{C}$ denote the Pareto frontier of the original Problem $\mathscr{P}1$, and let $\mathcal{C}_{g}$ denote the Pareto frontier of the transformed Problem $\mathscr{P}_1^\mathrm{v}$. The following theorem establishes that, under the mild assumption that $\mathbf{g}(\cdot)=(g_{1}(\cdot),g_{2}(\cdot))$ is component-wise strictly increasing, the Pareto frontier of $\mathscr{P}_1^\mathrm{v}$ can be readily obtained without resolving the optimization problem. 

\begin{theorem}[Order-preserving equivalence]\label{thm:Frontier_transformation}
If $\mathbf{g}(\cdot)=(g_{1}(\cdot),g_{2}(\cdot))$ is component-wise strictly increasing over the feasible ranges of
$\theta$ and $E$, then $\mathscr{P}_1^\mathrm{v}$ and $\mathscr{P}_1$ share the same set of Pareto-optimal policies, and
\begin{equation*}
\mathcal{C}_{g}=\mathbf{g}\left(\mathcal{C}\right)\triangleq\left\{ \left(g_{1}\left(\theta\right),g_{2}\left(E\right)\right):\left(\theta,E\right)\in\mathcal{C}\right\}.
\end{equation*}
\end{theorem}
\begin{IEEEproof}
Because $g_1(\cdot)$ and $g_2(\cdot)$ are strictly increasing,
$\boldsymbol{g}(\theta,E)=(g_1(\theta),g_2(E))$ preserves componentwise
dominance by the standard
order-isomorphism property \cite[Proposition~1.3.5]{Schr:B16}. Denote $\boldsymbol{f}(\cdot)=\{\theta(\cdot),E(\cdot)\}$, then, for
any feasible policies $\boldsymbol{\pi}$ and $\boldsymbol{\pi}'$,
$
\boldsymbol{f}\left(\boldsymbol{\pi}^{\prime}\right)\preceq\boldsymbol{f}\left(\boldsymbol{\pi}\right)\Longleftrightarrow\boldsymbol{g}\left(\boldsymbol{f}\left(\boldsymbol{\pi}^{\prime}\right)\right)\preceq\boldsymbol{g}\left(\boldsymbol{f}\left(\boldsymbol{\pi}\right)\right).
$
Thus, the Pareto-optimal policies are preserved, and the transformed frontier is
\begin{multline*}
\mathcal{C}_{g}=\left\{ \boldsymbol{g}\left(\boldsymbol{f}\left(\boldsymbol{\pi}\right)\right)|\boldsymbol{\pi}\in\Pi_{g}^{*}\right\} =\left\{ \boldsymbol{g}\left(\boldsymbol{f}\left(\boldsymbol{\pi}\right)\right)|\boldsymbol{\pi}\in\Pi_{f}^{*}\right\} \\
=\left\{ \left(g_{1}\left(\theta\right),g_{2}\left(E\right)\right):\left(\theta,E\right)\in\mathcal{C}\right\} .
\end{multline*}
\end{IEEEproof}

\begin{figure*}
\begin{centering}
\includegraphics[width=\textwidth]{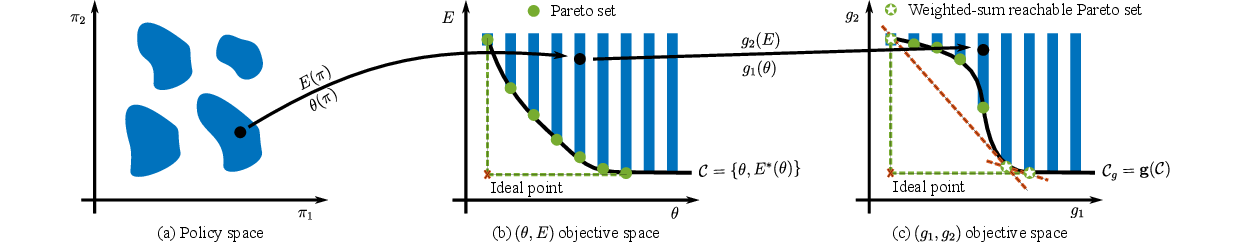}
\par\end{centering}
\caption{\label{fig:Pareto_illu}Illustration of Pareto optimality for a two-variable,
two-objective optimization problem example. (a) The nonconvex policy
space. (b) Bi-objective $(E,\theta)$ space, where the feasible space
and the objective space are non-convex and non-continuous, thereby
classical heuristic algorithms are challenging to find the complete
Pareto frontier. (c) Bi-objective $(g_{2},g_{1})$ space, where the
Pareto frontier is non-concave; therefore, the weighting method cannot
guarantee the discovery of all Pareto optima. Conversely, one can
map a point on the transformed Pareto frontier back to the original
frontier and then to the policy space, thereby obtaining an optimal
policy for the corresponding variant problems.}
\end{figure*}

\subsection{Preference-based Scalarization\label{subsec:Preference-based-scalarization}}
For preset user preferences, the bi-objective Pareto optimization reduces to a single-objective problem. Let $u:\mathbb{R}^{2}\to\mathbb{R}$ be a nonlinear increasing function that aggregates the objective vector $(\theta,E)$ into a single measure. The resulting scalarized optimization problem is formulated as
\begin{align*}
\mathscr{P}^\mathrm{v}_2:\underset{\boldsymbol{\pi}\in\Pi}{\text{min}} & \ u\left(\mathbf{g}\left(\theta\left(\boldsymbol{\pi}\right),E\left(\boldsymbol{\pi}\right)\right)\right),\text{s.t. }\tau\left[t\right]\le\bar{\tau},\forall t\in\mathcal{T}.
\end{align*}

The relationship between the optimal solution of this scalarized problem and the Pareto frontier of the original multi-objective problem is established in the following.

\begin{theorem}[{Sufficiency of monotone utility~\cite[Theorem 2.6.2]{Kai:B98}}]\label{thm:value_function}
If the value function $u\left(\cdot\right)$ is strongly increasing,
\emph{i.e.}, for any $\mathbf{x},\mathbf{y}\in\mathbb{R}^{2}$ with
$\mathbf{x}\preceq\mathbf{y}$, it holds that $u\left(\mathbf{x}\right)<u\left(\mathbf{y}\right)$,
then any solution to $\mathscr{P}^\mathrm{v}_2$ lies in the
Pareto frontier of $\mathscr{P}1$.
\end{theorem}

A canonical example satisfying this monotonicity condition is the weighted $L_{p}$ metric, defined as 
\begin{equation*}
u=\left(\alpha\left|\theta-\theta^{\prime}\right|^{p}+\left(1-\alpha\right)\left|E-E^{\prime}\right|^{p}\right)^{1/p},\alpha\in\left[0,1\right],p\ge1.
\end{equation*}
Here, the weight $\alpha$ explicitly quantifies the trade-off preference between the {\ac{rb} consumption} and energy consumption, and $\left(\theta^{\prime},E^{\prime}\right)$
is a special target. The metric becomes a weighted-sum metric when $p=1$ and $\theta^{\prime}=E^{\prime}=0$.

\subsection{Budget-Constrained Optimization}

This formulation focuses on a single primary objective while treating the secondary objective as a constraint (referred to as the $\varepsilon$-constraint method). When the energy consumption is the primary concern, the goal is to minimize the energy cost $E$ or $g_{2}\left(E\right)$ subject to an admissible upper bound $\varepsilon_{\theta}$ on the spatiotemporal {\ac{rb} consumption}
\begin{align*}
\mathscr{P}_\mathrm{v}^3:\underset{\boldsymbol{\pi}\in\Pi}{\text{min}} & \ g_{2}\left(E\left(\boldsymbol{\pi}\right)\right),\text{s.t. }\tau\left[t\right]\le\bar{\tau},\forall t\text{ and }g_{1}\left(\theta\left(\boldsymbol{\pi}\right)\right)\le\varepsilon_{\theta}.
\end{align*}

Similarly, one can formulate a {spatiotemporal \ac{rb} consumption} minimization problem with an energy constraint.

\begin{theorem}[{Sufficiency of partial optimization~\cite[Theorem 2.10.3]{Kai:B98}}]
\label{thm:partial_optimization}
Any optimal solution to $\mathscr{P}_\mathrm{v}^3$ (where the constraint
is active) is Pareto optimal for the original problem $\mathscr{P}1$.
\end{theorem}

The established theorems provide a transfer rule for system optimization, as illustrated in Fig.~\ref{fig:Pareto_illu}. Once the Pareto frontier $\mathcal{C}$ of $\mathscr{P}1$ is computed, the solutions for the transformed-objective problem $\mathscr{P}_\mathrm{v}^1$, the preference-based single-objective formulation $\mathscr{P}_\mathrm{v}^2$, and the budget-constrained single-objective formulation $\mathscr{P}_\mathrm{v}^3$ can be obtained through direct mapping or by searching over the frontier $\mathcal{C}$. In other words, \emph{we can recover the solutions to these variants without re-solving the problems.} This offers two main benefits. First, we bypass the analytical and computational difficulties of directly solving the variants. When the variants are non-convex or intractable, these transfer rules enable efficient recovery of the optimal solutions. Second, it shows strong potential for emergency response and policy adaptation in safety-critical applications.


\section{Simulation Results\label{sec:Simulation}}
We consider a \ac{uav}-based patrol system, where a \ac{uav} follows a circular trajectory to monitor a $200\times200\,m^{2}$ area. The \ac{uav} operates at an altitude of $50$ m with a flight speed of 6 m/s. On the ground, $N$ \acpl{bs} are deployed, with their positions randomly generated.

The channel gain is realized by $h_{n}\left[k,t\right]=g_{n}\left[k,t\right]\zeta_{n}\left[k,t\right]$ according to \eqref{eq:channel_model}. Specifically, the shape parameters $\kappa_{n}\left[k,t\right]$ of Gamma distribution of small-scale fading $\zeta^{sm}_{n}\left[k,t\right]$ are set randomly in {$[1,15]$, and the prediction error is set $\sigma_n[k,t]=0$}. As in \cite{LiChe:J24b}, the large-scale fading $g_{n}\left[k,t\right]$ includes path loss and shadowing, where the path loss is generated
by the 3GPP {Urban Micro} model \cite{TR36814}. The channel block state is generated by the \ac{los} probability model \cite{MozSadBen:J17}. In contrast, shadowing is generated as a spatially correlated log-normal process. The default implementation parameters are listed in Table \ref{tab:Implementation settings}.

\begin{table}[t]
\caption{\label{tab:Implementation settings}Default Implementation Parameters}
\centering{}%
\begin{tabular}{>{\raggedright}p{0.3\columnwidth}|>{\raggedright}p{0.6\linewidth}}
\hline 
\textbf{Parameter} & \textbf{Description}\tabularnewline
\hline 
\rowcolor{lightcyan}Patrol \ac{uav} trajectory & Circular trajectory at 50\,m altitude.\tabularnewline
\ac{uav} speed & Uniform in $6$ m/s.\tabularnewline
\rowcolor{lightcyan}Base station location & Uniformly distributed at ground level (0\,m).\tabularnewline
Carrier frequency & $f_{\text{c}}=3\text{ GHz}$.\tabularnewline
\rowcolor{lightcyan}Bandwidth & $B=10\text{ MHz}$.\tabularnewline
Noise power & $\delta^{2}=-90\text{ dBm}$.\tabularnewline
\rowcolor{lightcyan}Path loss (LOS) & $22.0+28.0\log_{10}(d)+20\log_{10}(f_{\text{c}})$.\tabularnewline
Path loss (NLOS) & $22.7+36.7\log_{10}(d)+26\log_{10}(f_{\text{c}})$.\tabularnewline
\rowcolor{lightcyan}LOS probability & $\mathbb{P}(\text{LOS},\theta)=(1+6\times\text{exp}(-0.15[\theta-6])^{-1}$,
where $\theta$ is elevation angle.\tabularnewline
Shadowing & Log-normal distribution with 0\,dB mean, 8\,dB variance, and 5\,m
correlation distance.\tabularnewline
\hline 
\end{tabular}
\end{table}

\subsection{Performance on Power and Access Control}
{As discussed in Section~\ref{subsec:alg_analysis}, the only potential optimality loss arises from binary recovery and the expected-capacity approximation for the proposed algorithm.} Fig.~\ref{fig:algorithm_1_analysis} {and Table~\ref {tab:complexity} evaluate the optimality, feasibility, and computational complexity of (exact and approximate)} Algorithm~\ref{alg:opt_power_control}
{over} 100 {independent random realizations}. As a benchmark, CVX~\cite{cvx} {solves the relaxation of $\mathscr{P}4$ using the explicit capacity lower bound, providing an objective lower bound on the corresponding integer approximate problem}.

Table~\ref {tab:complexity} shows that the proposed algorithm scales empirically as $\mathcal{O}(K^{1.3})$ over the tested range, below its theoretical bound $\mathcal{O}(K^2)$. It is over 200 times faster than the CVX baseline, and 600 times faster than the exact algorithm. A generic solver would take even longer to solve the exact formulation because it cannot exploit the derived structure and must repeatedly evaluate the expectation numerically.

Fig.~\ref {fig:Optimality_analysis_br} shows that the recovered binary solution nearly attains the CVX lower bound across all tested settings. It remains feasible in every case, whereas the fraction of nonbinary variables in the relaxed CVX solution reaches 50\% at $K=150$, as shown in Fig.~\ref{fig:Feasibility_analysis}. To isolate the effect of the expected-capacity approximation, Fig.~\ref{fig:Optimality_analysis_cr} compares the approximate and exact algorithms. The approximation increases energy consumption by only about 12\% for $\kappa\in[1,5]$ (low-\ac{los}), while this gap becomes negligible for $\kappa\in[11,15]$ (high-\ac{los}), consistent with its asymptotic optimality.

\subsection{Performance on Sampling Timing Control}
{We compare MPComm with three baselines. To isolate the effects of prediction and joint control of data generation and transmission, all schemes use the proposed power- and access-control algorithm but adopt different prediction horizons and sampling policies.}
\begin{itemize}
\item {\emph{Myopic-rate \cite{MatShe:J25}:} Uses only one-step channel information and serves the traffic independently in each slot under} $c_{m}(t)\ge S/\bar{\tau}$, $\forall t$.
\item {\emph{Predictive-aggregate \cite{ZhaLiRonZen:J25}:} Uses long-horizon channel predictions but assumes full-buffer traffic and constrains only the aggregate rate}, \emph{i.e.}, $\sum_{t}c_{m}(t)/T\ge S/\bar{\tau}$.
\item {\emph{Predictive-periodic:} Uses the same long-horizon predictions and Peak AoI requirement as MPComm but fixes the sampling intervals as $t_{i+1}-t_{i}=\bar{\tau}$.}
\end{itemize}
\begin{table}[t]
\caption{\label{tab:complexity}Computational time (s) over slot numbers $T$ and \ac{rb} number $K$.}
\centering{}%
\begin{tabular*}{\columnwidth}
{@{\extracolsep{\fill}}l|ccccc@{}}
\hline 
$K$ ($T=2$) & 10 & 50 & 100 & 150 & 300\tabularnewline
\hline 
Alg.~\ref{alg:opt_power_control} (Approx.) & 0.034 & 0.095 & 0.247 & 0.439 & 1.347\tabularnewline
CVX (Approx.) & 4.992 & 24.64 & 69.27 & 97.43 & /\tabularnewline
Alg.~\ref{alg:opt_power_control} (Exact) & 14.32 & 52.33 & 86.40 & / & /\tabularnewline
\hline 
\hline 
$K$ ($T=4$) & 10 & 50 & 100 & 150 & 300\tabularnewline
\hline 
Alg.~\ref{alg:opt_power_control} (Approx.) & 0.046 & 0.232 & 0.600 & 1.080 & 3.283\tabularnewline
CVX (Approx.) & 9.718 & 80.32 & / & / & /\tabularnewline
Alg.~\ref{alg:opt_power_control} (Exact) & 33.73 & / & / & / & /\tabularnewline
\hline 
\end{tabular*}

\vspace{1mm}
\parbox{\columnwidth}{\footnotesize
\textit{Note:} '/' indicates that the computation did not finish within 100 seconds.
}
\end{table}
\begin{figure*}[t]
    \centering
    \subfloat[{Optimality over binary recovery.}\label{fig:Optimality_analysis_br}]{
        \includegraphics[width=0.33\textwidth]{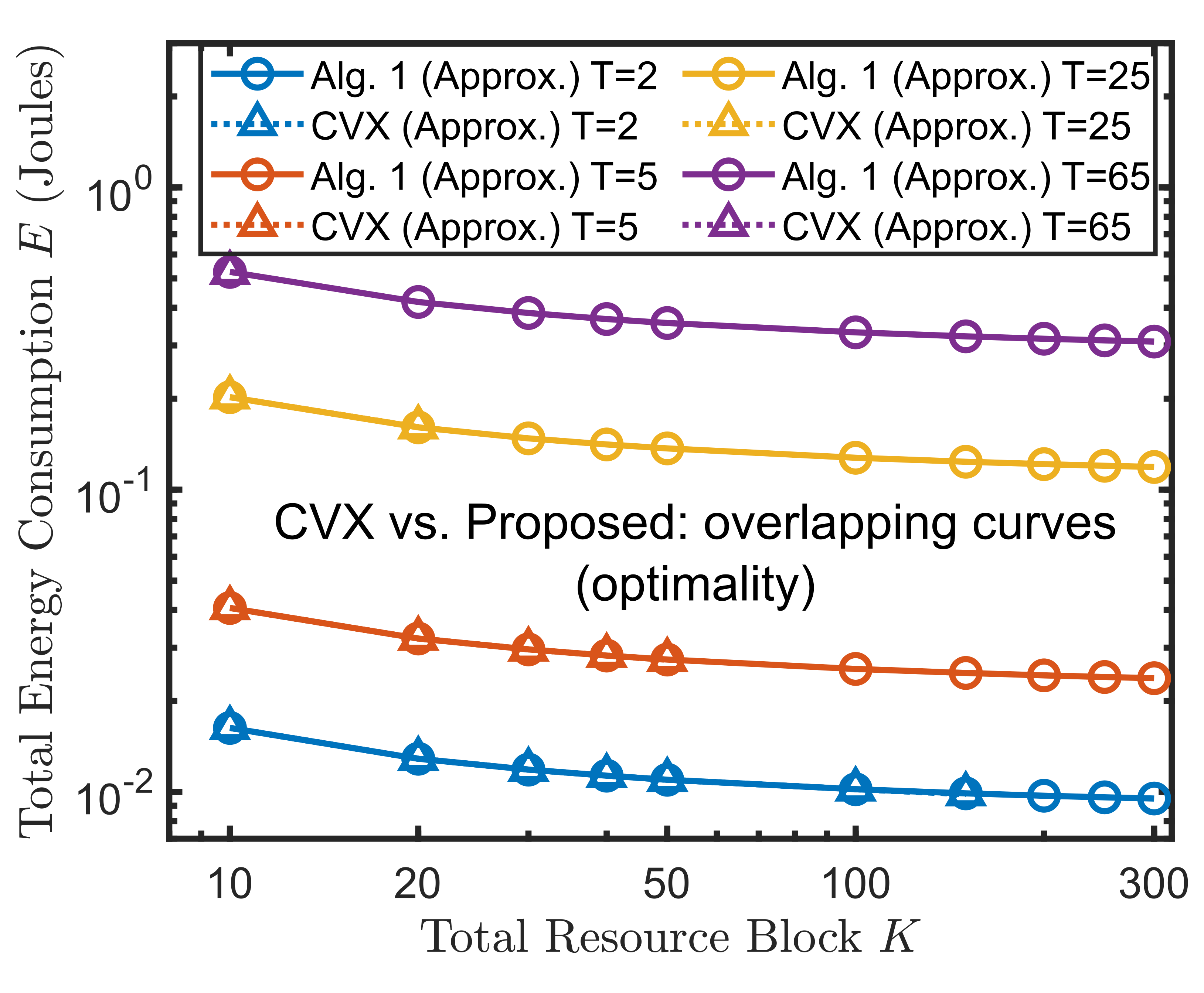}
    }
    \subfloat[Feasibility analysis.\label{fig:Feasibility_analysis}]{
        \includegraphics[width=0.33\textwidth]{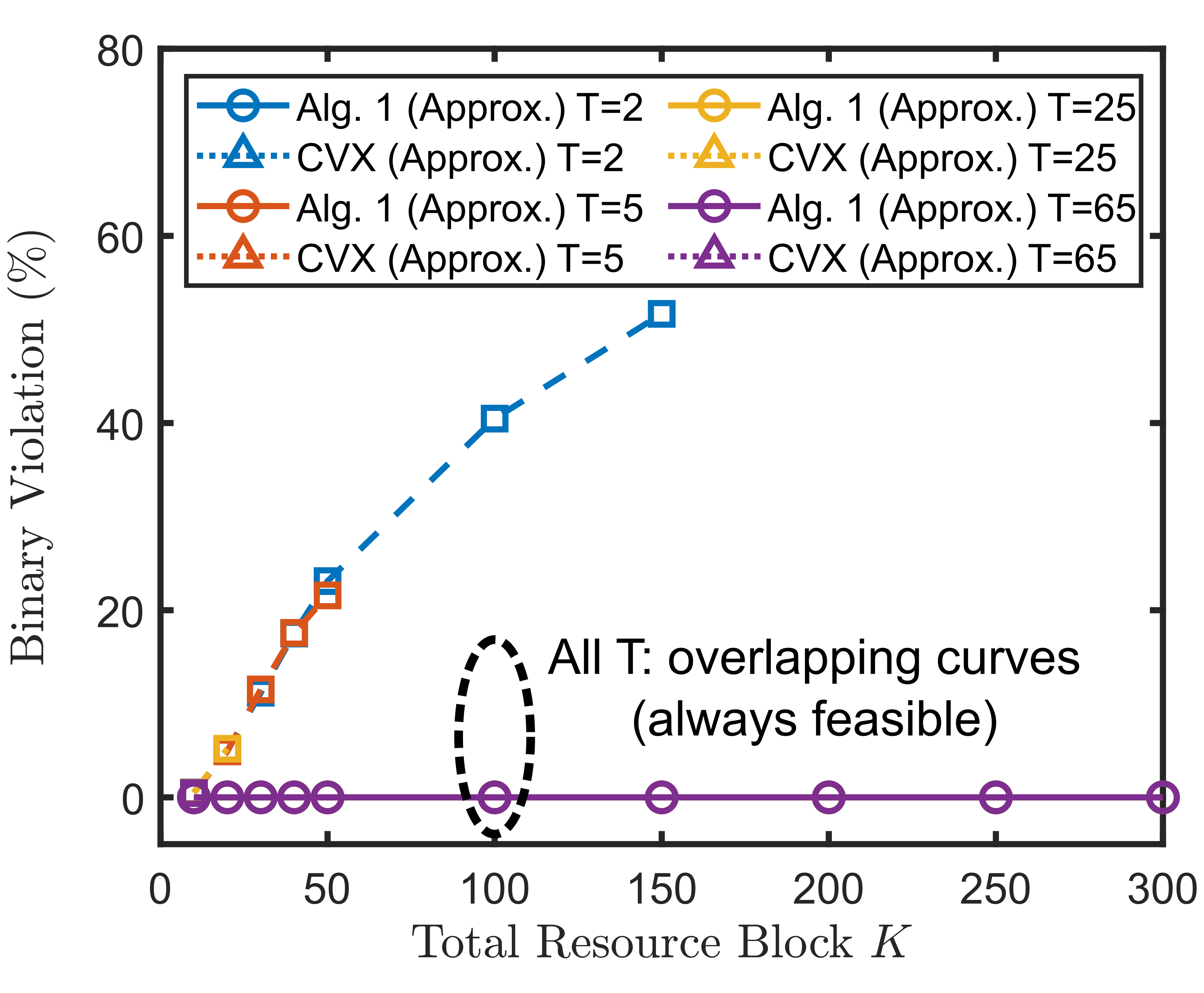}
    }
    \subfloat[{Optimality over lower-bound capacity.}\label{fig:Optimality_analysis_cr}]{
        \includegraphics[width=0.33\textwidth]{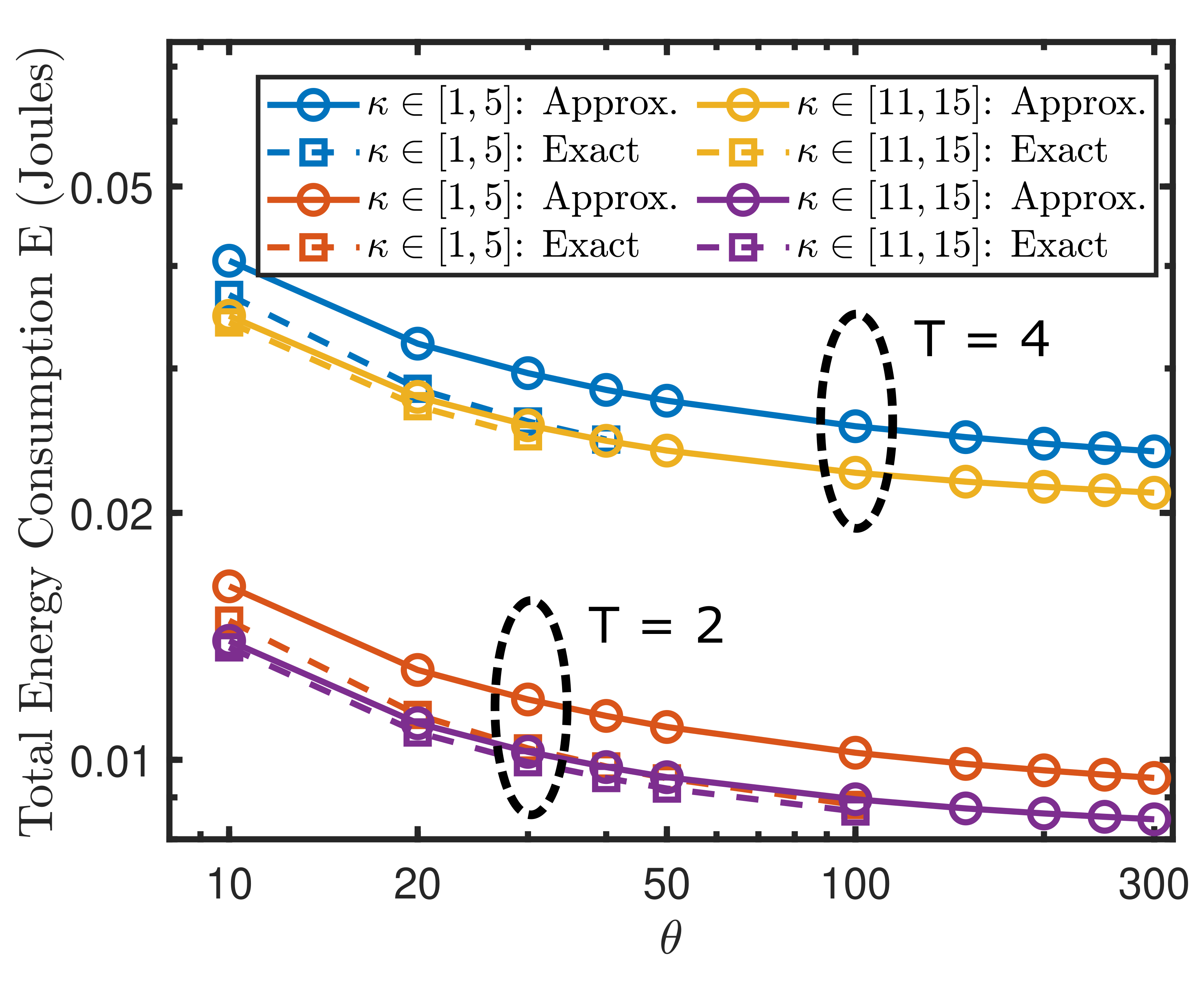}
    }
    \caption{Optimality and feasibility of Algorithm \ref{alg:opt_power_control} versus $K$ and $T$ for $N=5$.}
    \label{fig:algorithm_1_analysis}
\end{figure*}
\begin{figure*}[t]
    \centering
    \subfloat[{Pareto frontier.}\label{fig:pareto_opt}]{
        \includegraphics[width=0.33\textwidth]{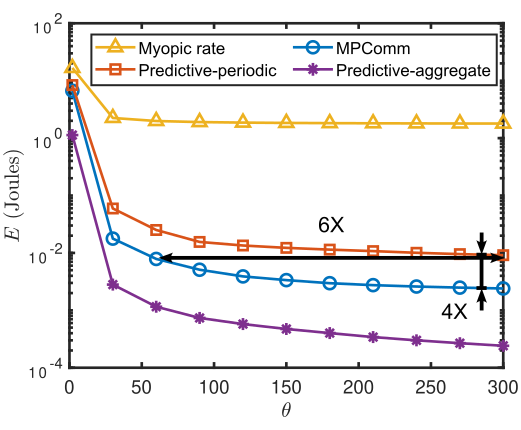}
    }
    \subfloat[{Success rate.}\label{fig:AOI_over_theta}]{
        \includegraphics[width=0.33\textwidth]{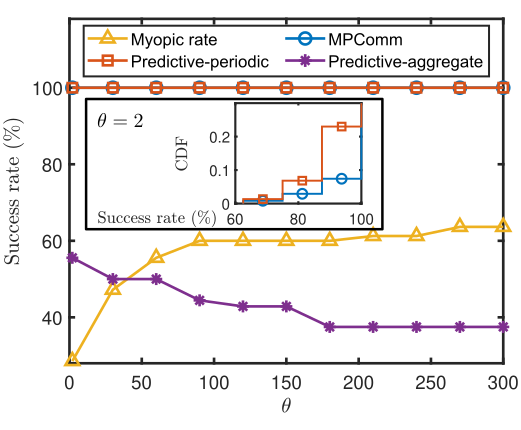}
    }
    \subfloat[{Stability.}\label{fig:Theta_over_T}]{
        \includegraphics[width=0.33\textwidth]{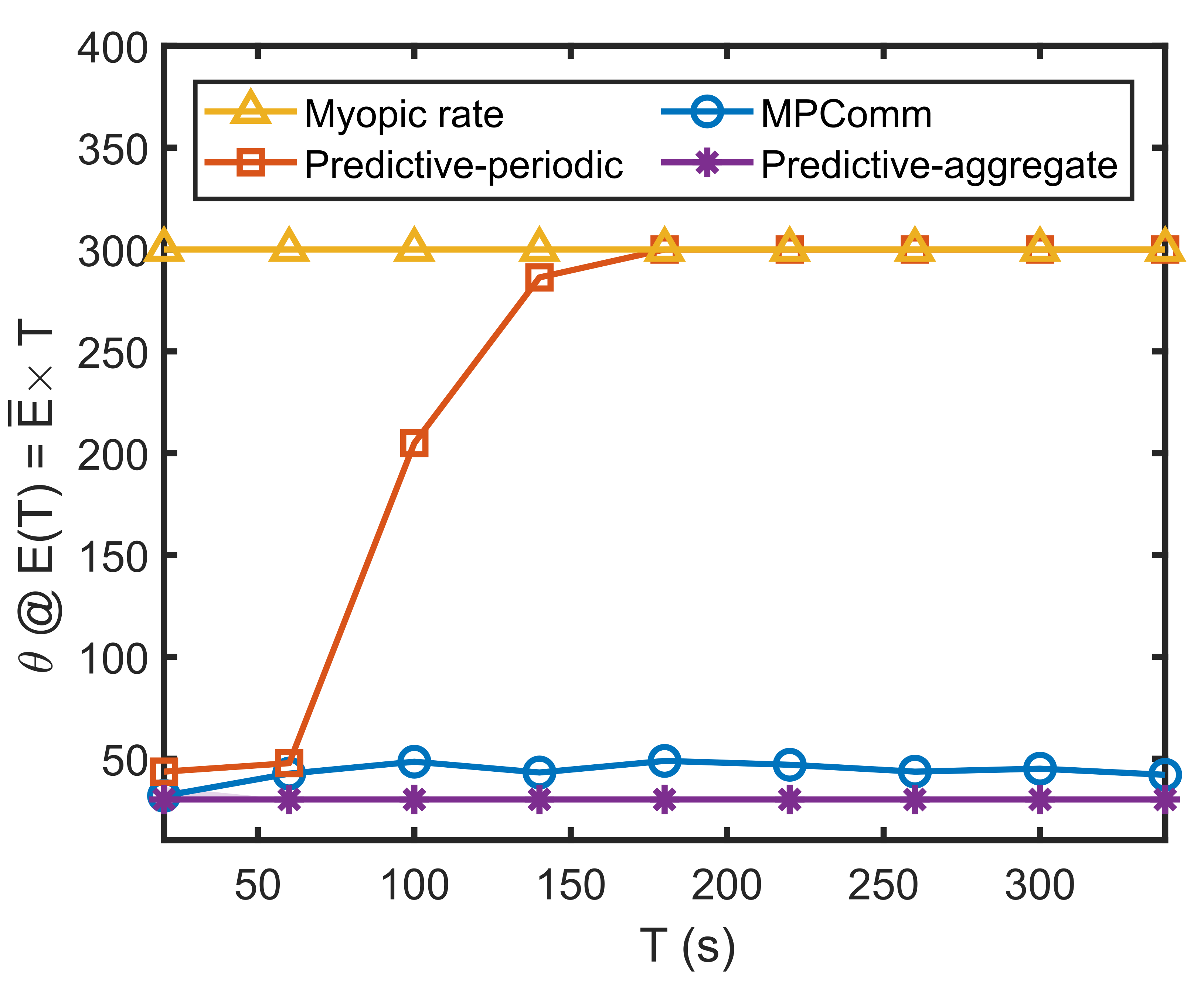}
    }
    \caption{Pareto frontier, success rate, and stability analysis for Algorithm~\ref{alg:g_control_alg} with $N=5$.}
    \label{fig:Whole_Performance}
\end{figure*}

Figs.~\ref{fig:pareto_opt} and \ref{fig:AOI_over_theta} demonstrate the benefit of predictive communication. Without predictive channel information, Myopic-rate cannot avoid unfavorable transmission intervals and therefore consumes over $20$ dB more energy than the predictive schemes while achieving a successful-service rate below $65\%$. By exploiting predicted channel variations, Predictive-periodic and MPComm schedule transmissions under favorable conditions and achieve nearly $100\%$ successful service.

The comparison between Predictive-periodic and MPComm isolates the gain from jointly controlling data generation and transmission. At $\theta=300$, MPComm reduces energy consumption by approximately four times. Conversely, at $E=10^{-2}$ J, it reduces the peak RB allocation from approximately $300$ to $50$, i.e., by a factor of six. Predictive-aggregate achieves lower energy and RB utilization by constraining only the horizon-wide aggregate rate. It may therefore concentrate resources on favorable slots or RBs without guaranteeing timely delivery of individual updates, resulting in only about $37\%$ successful service at $\theta=300$.

Fig.~\ref{fig:Theta_over_T} evaluates the stability of peak RB allocation under the fixed average-energy budget $E(T)=\bar{E}T$. MPComm maintains $\theta$ within approximately $40$-$50$ as $T$ increases, whereas Predictive-periodic progressively consumes more RBs and saturates at the maximum budget of $300$ for $T>200$ s. Thus, MPComm maintains stable resource utilization over long prediction horizons.

\section{Conclusion\label{sec:Conclusion}}
{In this paper, we developed MPComm to jointly optimize data generation and transmission for timely status updating in low-altitude networks. We devised an asymptotically optimal, low-complexity algorithm to characterize the Pareto frontier between \ac{uav} energy consumption and spatiotemporal \ac{rb} usage. Based on this frontier, the proposed transfer rules directly yield the frontiers or optimal solutions of transformed-objective, preference-based, and budget-constrained formulations, avoiding the need to solve each potentially nonconvex or intractable problem separately. Numerical results demonstrate near-optimal performance and more than $20$~dB energy gains from long-horizon prediction and $6$~dB from joint generation--transmission control, while maintaining a nearly $100\%$ service success rate and satisfying per-update timeliness constraints.}

{This study considered a single \ac{uav} following a predetermined trajectory under statistically characterized channel uncertainty. Promising extensions include interference-aware multi-\ac{uav} coordination, adaptive trajectory planning with optimality guarantees, and robust online operation with dynamically updated predictions.}

\appendices
\section{Proof of Proposition \ref{prop:pareto_front}\label{sec:proof_prop_pareto_front}}

We will prove that any point in $\mathcal{C}$ is Pareto-optimal for
Problem $\mathscr{P}1$ and all Pareto-optimal points of $\mathscr{P}1$
lie on $\mathcal{C}$. We first prove that $E^{*}(\varepsilon_{\theta})$ is monotonic over $\theta$.

\subsubsection{Monotonicity of $E^{*}(\varepsilon_{\theta})$} \label{subsubsec:monotonicity_of_E}
Denote the feasible set of Problem $\mathscr{P}2$ under  $\varepsilon_{\theta}$ as 
\begin{equation}
\mathcal{F}\left(\varepsilon_{\theta}\right)=\left\{ \boldsymbol{\pi}\in\Pi:\tau\left[t\right]\le\bar{\tau},\forall t\text{ and }\theta\left(\boldsymbol{\pi}\right)\le\varepsilon_{\theta}\right\} .\label{eq:def_f_s}
\end{equation}
Given any $\varepsilon_{\theta,1}<\varepsilon_{\theta,2}$, we have $\mathcal{F}(\varepsilon_{\theta,1})\subseteq\mathcal{F}(\varepsilon_{\theta,2})$ because
\begin{equation*}
\theta\left(\boldsymbol{\pi}\right)=\sum_{k\in\mathcal{K}}a_{n}\left[k,t\right]\le\varepsilon_{\theta,1}<\varepsilon_{\theta,2},\forall n,t.
\end{equation*}
Minimizing the same objective $E$ in $\mathscr{P}2$,
over a larger feasible set, yields a lower optimum. Therefore,
\begin{equation*}
E^{*}\left(\varepsilon_{\theta,1}\right)\ge E^{*}\left(\varepsilon_{\theta,2}\right),\forall\varepsilon_{\theta,1}<\varepsilon_{\theta,2},
\end{equation*}
showing that $E^{*}(\varepsilon_{\theta})$ is non-increasing over
$\varepsilon_{\theta}$.

Suppose that there exist $\underline{\theta}\le\varepsilon_{\theta,1}<\varepsilon_{\theta,2}\le\overline{\theta}$,
such that $E^{*}(\varepsilon_{\theta,1})=E^{*}(\varepsilon_{\theta,2})$.
Let $\boldsymbol{\pi}^{*}\in\mathcal{F}(\varepsilon_{\theta,1})$
be any optimal policy for Problem $\mathscr{P}2$ given the constraint
$\varepsilon_{\theta,1}$. By definition, $\theta(\boldsymbol{\pi}^{*})\le\varepsilon_{\theta,1}$
and $E(\boldsymbol{\pi}^{*})=E^{*}(\varepsilon_{\theta,1})$. Since
$\varepsilon_{\theta,1}<\varepsilon_{\theta,2}$, it follows $\mathcal{F}(\varepsilon_{\theta,1})\subseteq\mathcal{F}(\varepsilon_{\theta,2})$. Consequently, $\boldsymbol{\pi}^{*}$
is feasible for the problem with relaxed constraint $\varepsilon_{\theta,2}$.

Given the assumption $E^{*}(\varepsilon_{\theta,1})=E^{*}(\varepsilon_{\theta,2})$,
the policy $\boldsymbol{\pi}^{*}$ is also optimal for Problem $\mathscr{P}2$
under constraint $\varepsilon_{\theta,2}$, satisfying the resource
inequality strictly 
\begin{equation*}
\theta\left(\boldsymbol{\pi}^{*}\right)\le\varepsilon_{\theta,1}<\varepsilon_{\theta,2}.
\end{equation*}
The strict slackness of the constraint implies that the constraint $\theta\left(\boldsymbol{\pi}^{*}\right)<\varepsilon_{\theta,2}$
is inactive. From the principles of multiobjective optimization \cite[Theorem 3.2.2]{Kai:B98},
if relaxing a constraint yields no improvement in the objective function,
the objective must have reached its local unconstrained minimum with
respect to that resource. Thus, the system is saturated at $\varepsilon_{\theta,1}$,
implying 
\begin{equation*}
E^{*}\left(\varepsilon_{\theta,1}\right)=E^{\star}\Rightarrow\overline{\theta}\le\varepsilon_{\theta,1},
\end{equation*}
which contradicts the initial assumption that $\varepsilon_{\theta,1}<\varepsilon_{\theta,2}\le\overline{\theta}$
(specifically $\varepsilon_{\theta,1}<\overline{\theta}$).

Therefore, the assumption is false, and we conclude that $E^{*}\left(\varepsilon_{\theta,1}\right)>E^{*}\left(\varepsilon_{\theta,2}\right)$
for any $\varepsilon_{\theta,1}<\varepsilon_{\theta,2}$ with $\varepsilon_{\theta,1},\varepsilon_{\theta,2}\in[\underline{\theta},\overline{\theta}]\cap\mathbb{Z}_{+}$.

\subsubsection{Points on $\mathcal{C}$ are Pareto-optimal}
We prove this by contradiction. Suppose that a point $(\varepsilon_{\theta},E^{*}(\varepsilon_{\theta}))$
with $\varepsilon_{\theta}\in[\underline{\theta},\overline{\theta}]\cap\mathbb{Z}_{+}$
is not Pareto-optimal. This implies the existence of a feasible pair
$(\theta^{\prime},E^{\prime})$ that dominates $(\varepsilon_{\theta},E^{*}(\varepsilon_{\theta}))$,
meaning 
\begin{equation*}
\theta^{\prime}\le\varepsilon_{\theta},E^{\prime}\le E^{*}\left(\varepsilon_{\theta}\right),\text{ and }\left(\varepsilon_{\theta},E^{*}\left(\varepsilon_{\theta}\right)\right)\neq\left(\theta^{\prime},E^{\prime}\right).
\end{equation*}

First, by the definition of the optimal value function $E^{*}(\cdot)$,
any feasible pair must satisfy $E^{\prime}\ge E^{*}\left(\theta^{\prime}\right)$.
Combining this with the dominance condition yields 
\begin{equation}
E^{*}\left(\theta^{\prime}\right)\le E^{\prime}\le E^{*}\left(\varepsilon_{\theta}\right).\label{eq:ineq_E}
\end{equation}

We consider two cases for $\theta^{\prime}$. Case $\theta^{\prime}<\varepsilon_{\theta}$:
Since $\varepsilon_{\theta}\le\overline{\theta}$, strict monotonicity
from Appendix~\ref{subsubsec:monotonicity_of_E} implies $E^{*}\left(\theta^{\prime}\right)>E^{*}\left(\varepsilon_{\theta}\right)$.
This directly contradicts inequality \eqref{eq:ineq_E}. Case $\theta^{\prime}=\varepsilon_{\theta}$:
Inequality \eqref{eq:ineq_E} implies $E^{*}\left(\varepsilon_{\theta}\right)\le E^{\prime}\le E^{*}\left(\varepsilon_{\theta}\right)$,
forcing $E^{\prime}=E^{*}\left(\varepsilon_{\theta}\right)$. Thus,
$(\varepsilon_{\theta},E^{*}(\varepsilon_{\theta}))=(\theta^{\prime},E^{\prime})$,
which contradicts the requirement that a dominating point must be
distinct from the original point.

Thus, we conclude that since no such dominating pair exists, every point on $\mathcal{C}$ is Pareto-optimal.

\subsubsection{Every Pareto-Optimal Feasible Pair Lies on $\mathcal{C}$}

Let $(\theta^{\prime},E^{\prime})$ be any Pareto-optimal point. We
will prove $(\theta^{\prime},E^{\prime})\in\mathcal{C}_{}$, that
is $E^{\prime}=E^{*}\left(\theta^{\prime}\right)$ and $\theta^{\prime}\le\bar{\theta}$.

Feasibility implies $E^{\prime}\ge E^{*}\left(\theta^{\prime}\right)$.
If $E^{\prime}>E^{*}\left(\theta^{\prime}\right)$, the point $(\theta^{\prime},E^{*}(\theta^{\prime}))$
is feasible and strictly dominate $(\theta^{\prime},E^{\prime})$.
Thus, Pareto optimality necessitates $E^{\prime}=E^{*}(\theta^{\prime})$.

Suppose $\theta^{\prime}>\bar{\theta}$. By the definition of the
saturation point $\bar{\theta}$ in \eqref{eq:def_theta_up}, we have
$E^{*}(\theta^{\prime})=E^{*}(\bar{\theta})$. It implies consuming strictly less resource for the same energy performance.
Thus, Pareto optimality necessitates $\theta^{\prime}\le\bar{\theta}$.

Combining these results, any Pareto-optimal point must satisfy $E^{\prime}=E^{*}(\theta^{\prime})$
and $\theta^{\prime}\le\bar{\theta}$, which proves the proposition.

\section{Proof of Theorem \ref{prop:opt_trans}\label{sec:proof_prop_opt_trans}}
We prove the equivalence in three steps: 1) for any fixed sampling sequence $\mathbf t$, the feasible set separates across intervals; 2) the sum of the interval-wise optima lower-bounds the global optimum; and 3) concatenating the interval-wise optimizers attains this bound.

\subsubsection{Conditional separability} For any $\mathbf t\in\Upsilon$, the horizon $\mathcal T$ is partitioned into disjoint intervals
\begin{equation*}
 \mathcal T_i\triangleq[t_i,t_{i+1})\cap\mathbb Z_+, \qquad \mathcal T=\bigcup_{i\in\mathcal I}\mathcal T_i.    
\end{equation*}
Define the local variables and energy as
\begin{equation*}
\mathbf x_i\triangleq (\boldsymbol{\mathcal A}_i,\boldsymbol{\mathcal P}_i),\, f_i(\mathbf x_i)\triangleq \sum_{t\in\mathcal T_i} \sum_{n\in\mathcal N,k\in\mathcal K} a_n[k,t]p_n[k,t].
\end{equation*}
The total energy is therefore $E=\sum_{i\in\mathcal I}f_i(\mathbf x_i)$.
For fixed $\mathbf t$, the service constraint for update $i$ involves only the variables over $\mathcal T_i$. The access and power constraints are imposed per slot, while the \ac{rb} consumption constraint
\begin{equation*}
\theta(\boldsymbol\pi)\leq\varepsilon_\theta \iff \sum_{k\in\mathcal K}a_n[k,t]\leq\varepsilon_\theta,\, \forall n\in\mathcal N,\ t\in\mathcal T,
\end{equation*}
and can therefore be imposed separately over each $\mathcal T_i$. Hence, no constraint couples different intervals, and the feasible set can be factorized as
\begin{equation*}
\mathcal F(\mathbf t) = \prod_{i\in\mathcal I}\mathcal F_i(t_i,t_{i+1}),
\end{equation*}
where
\begin{equation*}
\mathcal F_i(t_i,t_{i+1}) \triangleq \left\{ \mathbf x_i: \begin{array}{l} \mathbb E\{\upsilon(t_i,t_{i+1})\}\geq\bar{\upsilon},\\ \sum_{k\in\mathcal K}a_n[k,t]\leq\varepsilon_\theta,\ \forall n,\ t\in\mathcal T_i,\\ \boldsymbol{\mathcal A}_i\in\mathcal S_{\mathbf A}^{|\mathcal T_i|},\ \boldsymbol{\mathcal P}_i\in\mathcal S_{\mathbf P}^{|\mathcal T_i|} \end{array} \right\}.
\end{equation*}
Define the interval-wise optimal value as
\begin{equation*}
E^*(t_i,t_{i+1}) \triangleq \min_{\mathbf x_i\in\mathcal F_i(t_i,t_{i+1})} f_i(\mathbf x_i),
\end{equation*}
with $E^*(t_i,t_{i+1})=+\infty$ if the interval is infeasible.

\subsubsection{Lower bound on the global optimum} For any globally feasible policy $\mathbf x\in\mathcal F(\mathbf t)$, its restriction to each interval satisfies $ \mathbf x_i\in\mathcal F_i(t_i,t_{i+1})$.
Therefore, by the definition of $E^*(t_i,t_{i+1})$, we have $f_i(\mathbf x_i)\geq E^*(t_i,t_{i+1})$, for all $i\in\mathcal I$.
Summing over all intervals and minimizing over $\mathcal F(\mathbf t)$ gives
\begin{equation*}
\min_{\mathbf x\in\mathcal F(\mathbf t)} \sum_{i\in\mathcal I}f_i(\mathbf x_i) \geq \sum_{i\in\mathcal I}E^*(t_i,t_{i+1}).
\end{equation*}

\subsubsection{Attainability by concatenation} Conversely, let $\mathbf x_i^*$ be an optimizer of the $i$th interval. Because the global feasible set is the Cartesian product of the interval-wise feasible sets, concatenating $\{\mathbf x_i^*\}_{i\in\mathcal I}$ yields a globally feasible policy attaining $ \sum_{i\in\mathcal I}f_i(\mathbf x_i^*) = \sum_{i\in\mathcal I}E^*(t_i,t_{i+1})$. Therefore,
\begin{equation*}
\min_{\mathbf x\in\mathcal F(\mathbf t)} \sum_{i\in\mathcal I}f_i(\mathbf x_i) = \sum_{i\in\mathcal I}E^*(t_i,t_{i+1}).
\end{equation*}

Finally, optimizing over the sampling sequence yields
\begin{equation*}
\min_{\mathbf t\in\Upsilon} \sum_{i\in\mathcal I}E^*(t_i,t_{i+1}).
\end{equation*}

\section{Proof of Propositions \ref{prop:opt_p_a}, \ref{prop:Integrality}
and \ref{prop:property_cc}\label{sec:proof_prop_opt_p_a}}

We prove Proposition \ref{prop:opt_p_a}, \ref{prop:Integrality}
and \ref{prop:property_cc} via the \ac{kkt} conditions
of $\mathscr{P}4$, constructing an explicit primal-dual policy, and
verifying that it satisfies all \ac{kkt} conditions. Since $\mathscr{P}4$
is convex and satisfies Slater's condition, the \ac{kkt}
conditions are necessary and sufficient for optimality \cite[Section 5.5.3]{Boy:B04}.

\subsubsection{Necessary and Sufficient Optimality Conditions}

Introduce Lagrange multiplier $\lambda$, $\{\mu_{1}[t]\}_{t\in\mathcal{T}_{i}}$,
$\{\mu_{2}[n,t]\}_{n\in\mathcal{N},t\in\mathcal{T}_{i}}$, $\{\mu_{3}[k,t]\}_{k\in\mathcal{K},t\in\mathcal{T}_{i}}$,
and $\{\mu_{4}[n,k,t],\mu_{5}[n,k,t],\mu_{6}[n,k,t]\}_{n\in\mathcal{N},k\in\mathcal{K},t\in\mathcal{T}_{i}}$,
associated with the constraints in $\mathscr{P}4$.
The \ac{kkt} conditions for $\mathscr{P}4$ consist of Primal
and dual feasibility: Constraints of $\mathscr{P}4$ and non-negativity
of multipliers; Complementary slackness
\begin{equation}
\lambda\Big(\bar{\upsilon}-\sum_{n\in\mathcal{N},k\in\mathcal{K},t\in\mathcal{T}_{i}}\phi_{n}\left[k,t\right]\Big)=0,\label{eq:kkt_c1}
\end{equation}
\begin{multline}
\mu_{1}\left[t\right]\Big(\sum_{n\in\mathcal{N},k\in\mathcal{K}}
a_{n}\left[k,t\right]c_{n,k,t}^{-1}\Big(\frac{\phi_{n}\left[k,t\right]}{a_{n}\left[k,t\right]}\Big)
-\bar{p}\Big)=0,\\
\forall t\in\mathcal{T}_{i},\label{eq:kkt_c2}
\end{multline}
\begin{equation}
\mu_{2}\left[n,t\right]\left(\sum_{k\in\mathcal{K}}a_{n}\left[k,t\right]-\varepsilon_{\theta}\right)=0,\forall n\in\mathcal{N},t\in\mathcal{T}_{i},\label{eq:kkt_c_add}
\end{equation}
\begin{equation}
\mu_{3}\left[k,t\right]\left(\sum_{n\in\mathcal{N}}a_{n}\left[k,t\right]-1\right)=0,\forall k\in\mathcal{K},t\in\mathcal{T}_{i},\label{eq:kkt_c3}
\end{equation}
\begin{equation}
\mu_{4}\left[n,k,t\right]\left(-a_{n}\left[k,t\right]\right)=0,\forall n\in\mathcal{N},k\in\mathcal{K},t\in\mathcal{T}_{i},\label{eq:kkt_c4}
\end{equation}
\begin{equation}
\mu_{5}\left[n,k,t\right]\left(a_{n}\left[k,t\right]-1\right)=0,\forall n\in\mathcal{N},k\in\mathcal{K},t\in\mathcal{T}_{i},\label{eq:kkt_c5}
\end{equation}
\begin{equation}
\mu_{6}\left[n,k,t\right]\left(-\phi_{n}\left[k,t\right]\right)=0,\forall n\in\mathcal{N},k\in\mathcal{K},t\in\mathcal{T}_{i}.\label{eq:kkt_c6}
\end{equation}
Stationarity (Denote $\mathcal{L}$ as the Lagrangian function)
\begin{equation}
\frac{\partial\mathcal{L}}{\partial\phi_{n}\left[k,t\right]}=0,\forall n\in\mathcal{N},k\in\mathcal{K},t\in\mathcal{T}_{i},\label{eq:kkt_add_2}
\end{equation}
\begin{equation}
\frac{\partial\mathcal{L}}{\partial a_{n}\left[k,t\right]}=0,\forall n\in\mathcal{N},k\in\mathcal{K},t\in\mathcal{T}_{i}.\label{eq:kkt_add3}
\end{equation}

We next derive the optimal policy from these conditions.

\subsubsection{Optimal Spectral Efficiency\label{subsec:opt_p}}

Taking the partial derivative with respect to $\phi_{n}\left[k,t\right]$
in \eqref{eq:kkt_add_2}, and
based on the chain rule, we have 
\begin{equation}
\left(1+\mu_{1}\left[t\right]\right)
\left(\frac{\partial \bar{c}_n[k,t]}{\partial {p}_n[k,t]}\right)^{-1}
-\lambda-\mu_{6}\left[n,k,t\right]=0.\label{eq:kkt_c7}
\end{equation}
If $\bar{c}_{n}\left[k,t\right]>0$, then $a_{n}\left[k,t\right]>0$
according to the definition of $p_{n}\left[k,t\right]$ in \eqref{eq:def_p_over_a_phi}.
We have $\mu_{6}\left[n,k,t\right]=0$ according to \eqref{eq:kkt_c6},
then \eqref{eq:kkt_c7} becomes
\begin{equation}
    \mathbb{E}\left\{\frac{g_n[k,t]\zeta_n[k,t]}{\delta^2+p_n[k,t]g_n[k,t]\zeta_n[k,t]}\right\}=\frac{\ln 2\left(1+\mu_{1}[t]\right)}{\lambda}\triangleq \lambda^{-1}_t.
    \label{eq:a_f_of_p}
\end{equation}
Denote the zero point of \eqref{eq:a_f_of_p} as $\tilde{p}_n[k,t]$, then the optimal power allocation and capacity is 
\begin{equation}
p_{n}^{*}\left[k,t\right]=\left[\tilde{p}_n[k,t]\right]^{+},\label{eq:a_opt_p}
\end{equation}
\begin{equation}
\bar{c}_{n}^{*}[k,t]
=\mathbb{E}\left\{\log_2\left(1+\frac{p_{n}^{*}\left[k,t\right]g_{n}\left[k,t\right]\zeta_{n}\left[k,t\right]}{\delta^2}\right)\right\}.\label{eq:a_opt_c}
\end{equation}

\subsubsection{Optimal \ac{rb} Allocation\label{subsec:opt_a}}

Similarly, taking the partial derivative with respect to $a_{n}\left[k,t\right]$
in \eqref{eq:kkt_add3}
we have 
\begin{equation}
{H_{n}\left[k,t\right]}+\mu_{2}\left[n,t\right]+\mu_{3}\left[k,t\right]-\mu_{4}\left[n,k,t\right]+\mu_{5}\left[n,k,t\right]=0,\label{eq:kkt_c8}
\end{equation}
where 
\begin{align*}
{H_{n}\left[k,t\right]}&=\left(1+\mu_{1}\left[t\right]\right)\\
&\times\left( c^{-1}_{n,k,t}\left(\frac{\phi_n[k,t]}{a_n[k,t]}\right)-\frac{\phi_n[k,t]}{a_n[k,t]}\left(\frac{\partial \bar{c}_n[k,t]}{\partial {p}_n[k,t]}\right)^{-1}
\right) \\
&=\left(1+\mu_{1}\left[t\right]\right)
\left( p_n^*[k,t]-c_n^*[k,t]\ln 2\lambda_t
\right).
\end{align*}
The stationarity condition \eqref{eq:kkt_c8} with corresponding complementary
slackness conditions \eqref{eq:kkt_c_add}--\eqref{eq:kkt_c5}
can be equivalently viewed from a Lagrangian function $\mathcal{L}_{2}=\sum_{t\in\mathcal{T}_{i}}\mathcal{L}_{3}\left[t\right]$,
where $\mathcal{L}_{3}\left[t\right]$ is defined as 
\begin{align*}
\mathcal{L}_{3}\left[t\right]= & \sum_{n\in\mathcal{N},k\in\mathcal{K}}H_{n}\left[k,t\right]a_{n}\left[k,t\right]\\
 & +\sum_{n\in\mathcal{N}}\mu_{2}\left[n,t\right]\left(\sum_{k\in\mathcal{K}}a_{n}\left[k,t\right]-\varepsilon_{\theta}\right)\\
 & +\sum_{k\in\mathcal{K}}\mu_{3}\left[k,t\right]\left(\sum_{n\in\mathcal{N}}a_{n}\left[k,t\right]-1\right)\\
 & +\sum_{n\in\mathcal{N},k\in\mathcal{K}}\mu_{4}\left[n,k,t\right]\left(-a_{n}\left[k,t\right]\right)\\
 & +\sum_{n\in\mathcal{N},k\in\mathcal{K}}\mu_{5}\left[n,k,t\right]\left(a_{n}\left[k,t\right]-1\right).
\end{align*}
This Lagrangian function $\mathcal{L}_{3}\left[t\right]$ corresponds
to the following \ac{rb} allocation problem 
\begin{align}
\underset{\{a_{n}\left[k,t\right]\}_{n\in\mathcal{N},k\in\mathcal{K}}}{\text{minimize}} & \sum_{n\in\mathcal{N},k\in\mathcal{K}}H_{n}\left[k,t\right]a_{n}\left[k,t\right],\label{eq:s1_obj}\\
\text{subject to}\quad & \sum_{k\in\mathcal{K}}a_{n}\left[k,t\right]\le\varepsilon_{\theta},\forall n\in\mathcal{N},\label{eq:s1_c1}\\
 & \sum_{n\in\mathcal{N}}a_{n}\left[k,t\right]\le1,\forall k\in\mathcal{K},\label{eq:s1_c2}\\
 & a_{n}\left[k,t\right]\in\left[0,1\right],\forall n\in\mathcal{N},k\in\mathcal{K}.\label{eq:s1_c3}
\end{align}
Since $1+\mu_{1}\left[t\right]\ge1$,
denote $w_{n}\left[k,t\right]=p_{n}^{*}\left[k,t\right]-\ln2\lambda_{t}\bar{c}_{n}^{*}\left[k,t\right]$,
therefore, the optimal \ac{rb} allocation $\mathbf{A}^{*}\left[t;\lambda_{t}\right]$
is the solution to 
\begin{equation}
\underset{\{a_{n}\left[k,t\right]\}_{n\in\mathcal{N},k\in\mathcal{K}}}{\text{min}}\sum_{n\in\mathcal{N},k\in\mathcal{K}}w_{n}\left[k,t\right]a_{n}\left[k,t\right]\text{ s.t. }\text{\eqref{eq:s1_c1}--\eqref{eq:s1_c3}.}\label{eq:p_a}
\end{equation}

\subsubsection{Capped Water Level}

Denote $\tilde{\lambda}_{t}$ as the solution to $p^{*}\left[t;\Lambda\right]=\bar{p}$.
Then, complementary slackness \eqref{eq:kkt_c2} implies two regimes: 
\begin{itemize}
\item Inactive constraint ($p^{*}[t;\tilde{\lambda}_{t}]<\bar{p}$): It
means $\mu_{1}[t]=0$, yielding $\lambda_{t}=\lambda/\ln2$;
\item Active constraint ($p^{*}[t;\tilde{\lambda}_{t}]=\bar{p}$): It implies
$\lambda_{t}=\tilde{\lambda_{t}}$.
\end{itemize}
Equivalently, the optimal effective water level is the capped value
\begin{equation}
\lambda_{t}=\min\left\{ \frac{\lambda}{\ln2},\tilde{\lambda_{t}}\right\} .\label{eq:opt_lambda_form}
\end{equation}

\subsubsection{Optimal Lagrangian Parameter}
Since $\bar{c}_{n}^{*}\left[k,t;\lambda_{t}\right]a_{n}^{*}\left[k,t;\lambda_{t}\right]>0$
if and only if $\lambda_{t}>0$ ({\em i.e.}, $\lambda>0$), any strictly
positive rate requirement $\bar{\upsilon}>0$ implies $\lambda>0$.
Thus, the rate constraint is active according to \eqref{eq:kkt_c1},
and the optimal $\lambda$ is the unique solution to the equation
\begin{equation}
\sum_{n\in\mathcal{N},k\in\mathcal{K},t\in\mathcal{T}_{i}}\bar{c}_{n}^{*}\left[k,t;\lambda_{t}\right]a_{n}^{*}\left[k,t;\lambda_{t}\right]=\bar{\upsilon}.\label{eq:opt_lambda}
\end{equation}

\subsubsection{Optimal Binary \ac{rb} Policy\label{subsec:opt_a_01}}

We will prove that there exists an optimal binary \ac{rb} allocation
policy for any $\lambda_{t}$, {\em i.e.}, $a_{n}^{*}\left[k,t\right]\in\{0,1\}$
for all $n\in\mathcal{N}$ and $k\in\mathcal{K}$.

Stack the variables $\mathbf{A}\left[t\right]$ for any $t$ into
a vector 
\begin{multline*}
\mathbf{x}_{t}\triangleq[a_{1}\left[1,t\right],\cdots,a_{1}\left[K,t\right],a_{2}\left[1,t\right],\cdots,a_{2}\left[K,t\right]\\
,\cdots a_{N}\left[1,t\right],\cdots,a_{N}\left[k,t\right]]^{\text{T}}.
\end{multline*}
With this definition, problem \eqref{eq:p_a} can be written in standard
linear programming form as 
\begin{equation*}
\min_{\mathbf{x}_{t}}\,\boldsymbol{w}_{t}^{\text{T}}\mathbf{x}_{t},\,\text{s.t.}\,\mathbf{M}\mathbf{x}_{t}\le\mathbf{b},\mathbf{x}_{t}\ge0,
\end{equation*}
where vector $\boldsymbol{w}_{t}$ collects the coefficients of $w_{n}\left[k,t\right]$,
and the constraint matrix $\mathbf{M}$ and right-hand side vector
$\mathbf{b}$ are given by
\begin{equation*}
\mathbf{M}=\left[\begin{array}{c}
\mathbf{I}_{N}\otimes\mathbf{1}_{K}^{\text{T}}\\
\mathbf{1}_{N}^{\text{T}}\otimes\mathbf{I}_{K}
\end{array}\right],\,\mathbf{b}=\left[\begin{array}{c}
\varepsilon_{\theta}\mathbf{1}_{N}\\
\mathbf{1}_{K}
\end{array}\right].
\end{equation*}
Note that $\sum_{n\in\mathcal{N}}a_{n}\left[k,t\right]\le1$ and $a_{n}\left[k,t\right]\ge0$
imply $a_{n}\left[k,t\right]\le1$. Thus, the explicit upper bound
constraint is redundant.

Each column of $\mathbf{M}$ contains exactly two nonzero entries,
both equal to $1$, and the rows of $\mathbf{M}$ can be partitioned
into two disjoint sets corresponding to user constraints and resource-block
constraints, with each column having at most one nonzero entry in
each set. Hence, $\mathbf{M}$ is a totally unimodular matrix \cite{HofKru:B09}.
Moreover, since $\varepsilon_{\theta}\in\mathbb{Z}_{+}$, the vector
$\mathbf{b}$ is integral. By the Hoffman-Kruskal theorem \cite{HofKru:B09},
problem \eqref{eq:p_a} has integer optima, {\em i.e.}, $a_{n}^{*}\left[k,t\right]\in\{0,1\}$
for all $n\in\mathcal{N}$ and $k\in\mathcal{K}$.

\subsubsection{Optimality, Continuity and Monotonicity\label{subsec:monotonicity}}

Define the per-slot sum rate and the consumed power at slot $t$
as 
\begin{equation}
\Phi^{*}\left[t;\Lambda\right]\triangleq\sum_{n\in\mathcal{N},k\in\mathcal{K}}\bar{c}_{n}^{*}\left[k,t;\Lambda\right]a_{n}^{*}\left[k,t;\Lambda\right],\label{eq:def_phi_t}
\end{equation}
\begin{equation}
P^{*}\left[t;\Lambda\right]\triangleq\sum_{n\in\mathcal{N},k\in\mathcal{K}}p_{n}^{*}\left[k,t;\Lambda\right]a_{n}^{*}\left[k,t;\Lambda\right],\label{eq:def_p_t}
\end{equation}
where $\Lambda\ge0$. Since the optimal variables $\{\bar{c}_{n}^{*}\left[k,t;\Lambda\right],p_{n}^{*}\left[k,t;\Lambda\right],a_{n}^{*}\left[k,t;\Lambda\right]\}$
derived in Sections \ref{subsec:opt_p} and \ref{subsec:opt_a} are
determined by $\Lambda$, both $\phi^{*}\left[t;\Lambda\right]$ and
$p^{*}\left[t;\Lambda\right]$ can be viewed as functions of $\Lambda$.

We define the set of switching points where the binary allocation
changes as
\begin{equation*}
\mathcal{S}_{\text{s}}=\left\{ \Lambda_{0}\ge0;\mathbf{A}^{*}\left[t;\Lambda_{0}\right]\neq\lim_{\Lambda\to\Lambda_{0}}\mathbf{A}^{*}\left[t;\Lambda\right]\right\} ,
\end{equation*}
where $\mathbf{A}^{*}\left[t;\Lambda_{0}\right]$ is the integer solution
to \eqref{eq:p_a}. The remaining non-switching domain is defined
as
\begin{equation*}
\mathcal{S}_{\text{ns}}=\bigcup_{i\in\left\{ 0,\cdots,|\mathcal{S}_{\text{s}}|\right\} }s_{i},\,s_{i}=\left(\mathcal{S}_{\text{s}}\left[i\right],\mathcal{S}_{\text{s}}\left[i+1\right]\right),
\end{equation*}
where $\mathcal{S}_{\text{s}}\left[i\right]$ denotes the $i$th smallest
point in $\mathcal{S}_{\text{s}}\bigcup\left\{ 0,\infty\right\} $
. It is clear that the closure of these sets satisfies $\mathcal{S}_{\text{s}}\bigcup\mathcal{S}_{\text{ns}}=\mathbb{R}_{0+}$.

For the first segment $s_{0}$ (where $0\le\Lambda<\iota_{n}\left[k,t\right]$),
all \acpl{rb} are inactive, yielding $\phi^{*}\left[t;\Lambda\right]=0$
and $p^{*}\left[t;\Lambda\right]=0$. For each subsequent segment
$s_{i}$ with $i>1$, the \ac{rb} allocation $\mathbf{A}^{*}\left[t;\Lambda\right]$
remains constant. Consequently, both $\phi^{*}\left[t;\Lambda\right]$
and $p^{*}\left[t;\Lambda\right]$ are finite linear combinations
of $\left\{ \bar{c}_{n}^{*}\left[k,t;\Lambda\right]\right\} $ and
$\left\{ p_{n}^{*}\left[k,t;\Lambda\right]\right\} $ which are continuously
increasing in $\Lambda$. Thus, $\phi^{*}\left[t;\Lambda\right]$
and $p^{*}\left[t;\Lambda\right]$ are continuous and strictly increasing
over any $\Lambda\in s_{i}$.

Next, we analyze the behavior at the switching points. Since the problem \eqref{eq:p_a} has a continuous objective and a compact feasible
set, its optimal objective value is continuous in $\Lambda$. Evaluating
the optimality condition at a switching point $\Lambda$ (where $\Lambda^{-}=\Lambda=\Lambda^{+}$),
we have 
\begin{equation}
p^{*}\left[t;\Lambda^{-}\right]-\ln2\Lambda\phi^{*}\left[t;\Lambda^{-}\right]=p^{*}\left[t;\Lambda^{+}\right]-\ln2\Lambda\phi^{*}\left[t;\Lambda^{+}\right].\label{eq:continuity_c2}
\end{equation}
Since $\Lambda>0$, the change in power and sum rate must share
the same sign.

The optimality of the policy at $\Lambda^{-}$ implies it minimizes
the Lagrangian at that specific water level 
\begin{equation*}
p^{*}\left[t;\Lambda^{-}\right]-\ln2\Lambda^{-}\phi^{*}\left[t;\Lambda^{-}\right]\le p^{*}\left[t;\Lambda^{+}\right]-\ln2\Lambda^{-}\phi^{*}\left[t;\Lambda^{+}\right].
\end{equation*}
Similarly, the optimality of the policy at $\Lambda^{+}$ implies
\begin{equation*}
p^{*}\left[t;\Lambda^{+}\right]-\ln2\Lambda^{+}\phi^{*}\left[t;\Lambda^{+}\right]\le p^{*}\left[t;\Lambda^{-}\right]-\ln2\Lambda^{+}\phi^{*}\left[t;\Lambda^{-}\right].
\end{equation*}
Summing these two inequalities, we have 
\begin{equation*}
\left(\Lambda^{+}-\Lambda^{-}\right)\left(\phi^{*}\left[t;\Lambda^{-}\right]-\phi^{*}\left[t;\Lambda^{+}\right]\right)\le0.
\end{equation*}
Since $\Lambda^{+}>\Lambda^{-}$, we have $\phi^{*}\left[t;\Lambda^{-}\right]\le\phi^{*}\left[t;\Lambda^{+}\right]$,
then $p^{*}\left[t;\Lambda^{-}\right]\le p^{*}\left[t;\Lambda^{+}\right]$
according to \eqref{eq:continuity_c2}. Therefor, $\phi^{*}\left[t;\Lambda\right]$
and $p^{*}\left[t;\Lambda\right]$ are strictly increasing over $\Lambda>\iota_{n}\left[k,t\right]$.

To reach the per-slot sum rate in $[\phi^{*}[t;\Lambda^{-}],\phi^{*}[t;\Lambda^{+}]]$
and power in $[p^{*}[t;\Lambda^{-}],p^{*}[t;\Lambda^{+}]]$ if $\phi^{*}\left[t;\Lambda^{-}\right]\neq\phi^{*}\left[t;\Lambda^{+}\right]$
and $p^{*}\left[t;\Lambda^{-}\right]\neq p^{*}\left[t;\Lambda^{+}\right]$,
we define the convex combination of $\mathbf{A}^{*}\left[t;\Lambda^{-}\right]$
and $\mathbf{A}^{*}\left[t;\Lambda^{+}\right]$ as 
\begin{equation*}
\mathbf{A}^{\xi}\left[t;\Lambda\right]=\xi\mathbf{A}^{*}\left[t;\Lambda^{-}\right]+\left(1-\xi\right)\mathbf{A}^{*}\left[t;\Lambda^{+}\right],
\end{equation*}
and prove that $\mathbf{A}^{\xi}\left[t;\Lambda\right]$ is optimal
for any $\xi$.

Since $p^{*}\left[t\right]$ and $\phi^{*}\left[t\right]$ are linear
combinations of continuous terms $\left\{ \bar{c}_{n}^{*}\left[k,t;\Lambda\right]\right\} $
and $\left\{ p_{n}^{*}\left[k,t;\Lambda\right]\right\} $, thus 
\begin{equation}
P_{\text{ext}}\left[t;\Lambda,\xi\right]-\ln2\Lambda\phi^{*}\left[t;\Lambda,\xi\right]=p^{*}\left[t;\Lambda\right]-\ln2\Lambda\phi^{*}\left[t;\Lambda\right]\label{eq:opt_xi}
\end{equation}
where 
\begin{align*}
P_{\text{ext}}\left[t;\Lambda,\xi\right] & \triangleq\xi\sum_{n\in\mathcal{N},k\in\mathcal{K}}p_{n}^{*}\left[k,t;\Lambda\right]a_{n}^{*}\left[k,t;\Lambda^{-}\right]\\
 & \quad\left(1-\xi\right)\sum_{n\in\mathcal{N},k\in\mathcal{K}}p_{n}^{*}\left[k,t;\Lambda\right]a_{n}^{*}\left[k,t;\Lambda^{+}\right]\\
 & \overset{(a)}{=}\xi\sum_{n\in\mathcal{N},k\in\mathcal{K}}p_{n}^{*}\left[k,t;\Lambda^{-}\right]a_{n}^{*}\left[k,t;\Lambda^{-}\right]\\
 & \quad\left(1-\xi\right)\sum_{n\in\mathcal{N},k\in\mathcal{K}}p_{n}^{*}\left[k,t;\Lambda^{+}\right]a_{n}^{*}\left[k,t;\Lambda^{+}\right]\\
 & =\xi p^{*}\left[t;\Lambda^{-}\right]+\left(1-\xi\right)p^{*}\left[t;\Lambda^{+}\right].
\end{align*}
Here, (a) holds because $p_{n}^{*}\left[k,t;\Lambda\right]$ is continuous
over $\Lambda$. Similarly, 
\begin{equation*}
\Phi_{\text{ext}}\left[t;\Lambda,\xi\right]=\xi\phi^{*}\left[t;\Lambda^{-}\right]+\left(1-\xi\right)\phi^{*}\left[t;\Lambda^{+}\right].
\end{equation*}
Thus, $\mathbf{A}^{\xi}\left[t;\Lambda\right]$ is optimal to problem
\eqref{eq:p_a} because it has the same optimal value as \eqref{eq:opt_xi},
$P_{\text{ext}}\left[t;\Lambda,0\right]=p^{*}\left[t;\Lambda^{-}\right]$
and $P_{\text{ext}}\left[t;\Lambda,1\right]=p^{*}\left[t;\Lambda^{+}\right]$.

\section{Proof of Lemma \ref{lem:c_lb}\label{sec:proof_lem_c_lb}}

Define a new random variable $Y=\ln(\gamma_{n}[k,t])$, such that
$\gamma_{n}[k,t]=e^{Y}$. Substituting this into the capacity equation
\begin{equation*}
\mathbb{E}\left\{ c_{n}\left[k,t\right]\right\} =\mathbb{E}_{Y}\left\{ \log_{2}\left(1+e^{Y}\right)\right\} .
\end{equation*}
Since function $f(y)=\log_{2}(1+e^{y})$ is strictly convex,
applying Jensen's inequality to the capacity equation, we have 
\begin{equation*}
\mathbb{E}_{Y}\left\{ \log_{2}\left(1+e^{Y}\right)\right\} \ge\log_{2}\left(1+e^{\mathbb{E}\left\{ Y\right\} }\right).
\end{equation*}

According to the definition of $\gamma_{n}\left[k,t\right]$ in \eqref{eq:def_snr}
and $h_{n}\left[k,t\right]$ in \eqref{eq:channel_model}, 
we have
\begin{align*}
\mathbb{E}\left\{ Y\right\}
&=
\ln\left(
\bar{\gamma}_{n}\left[k,t\right]
\right)
+\mathbb{E}\left\{\ln\zeta^\mathrm{pe}_n[k,t]\right\}
+\mathbb{E}\left\{\ln\zeta^\mathrm{sf}_n[k,t]\right\}\\
&=\ln\left(
\bar{\gamma}_{n}\left[k,t\right]
\right)
-\frac{\sigma^2_n}{2}
+\psi\left(\kappa_n[k,t]\right)
-\ln\kappa_n[k,t],
\end{align*}
where $\bar{\gamma}_{n}\left[k,t\right]\triangleq {p_{n}\left[k,t\right]g_{n}\left[k,t\right]}/{\delta^{2}}$. Denote $$\beta_n[k,t]=e^{\psi\left(\kappa_n[k,t]\right)-{\sigma^2_n}/{2}}/\kappa_n[k,t],$$ then $e^{\mathbb{E}\left\{ Y\right\} }=\beta_n[k,t]\bar{\gamma}_{n}\left[k,t\right]$ and
\begin{equation}
\mathbb{E}\left\{ c_{n}\left[k,t\right]\right\} \ge\log_{2}\left(1+\beta_{n}\left[k,t\right]\bar{\gamma}_{n}\left[k,t\right]\right)\triangleq \underbar{c}_{n}\left[k,t\right].\label{eq:lb_a}
\end{equation}

Next, we analyze the tightness of the lower bound. Clearly, as $\kappa_n[k,t]\to\infty$ {and $\sigma_n[k,t] \to 0$}, {$\beta_n[k,t]\to 1$.}
Then, \eqref{eq:lb_a} becomes
\begin{equation}
\mathbb{E}\left\{ c_{n}\left[k,t\right]\right\} \ge\log_{2}\left(1+\bar{\gamma}_{n}\left[k,t\right]\right).\label{eq:lb_k_infty}
\end{equation}
Moreover, {using the concavity} of $\log_2(1+x)$ gives
\begin{equation}
\mathbb{E}\left\{ c_n[k,t]\right\} \le \log_2\left(1+\mathbb{E}\{\gamma_n[k,t]\}\right)
=\log_2\left(1+\bar\gamma_n[k,t]\right). \label{eq:up_k}
\end{equation}
Combining \eqref{eq:lb_k_infty} and \eqref{eq:up_k}, we have
\begin{equation*}
\lim_{\kappa_n[k,t]\to\infty,\sigma_n[k,t] \to 0}
\left(\mathbb{E}\left\{ c_n[k,t]\right\}-\underbar{c}_n[k,t]\right)=0.
\end{equation*}

As $\bar\gamma_n[k,t]\to\infty$, we have
\begin{align*}
\mathbb{E} \left\{\log_2 \left(1+\gamma_n[k,t]\right)\right\}
&=
\mathbb{E} \left\{\log_2 \left(\gamma_n[k,t]\right)\right\}+o(1)\\
&=
\mathbb{E} \left\{\log_2 \left(\beta_n[k,t]\bar{\gamma}_n[k,t]\right)\right\}+o(1).
\end{align*}
Then, combining with the lower bound in \eqref{eq:lb_a}, we have
\begin{equation*}
\lim_{\bar\gamma_n[k,t]\to\infty}
\left(\mathbb{E}\left\{ c_n[k,t]\right\}-\underbar{c}_n[k,t]\right)=0.
\end{equation*}

\bibliographystyle{IEEEtran}
\bibliography{bib_files/abrv, bib_files/StringDefinitions, bib_files/BL}

\end{document}

%% file: bib_files/BL_acronym.tex
\newac{speb}{SPEB}{square position error bound}
\newac[plural=EFIMs,firstplural=Fisher information matrices (EFIMs)]{efim}{EFIM}{Fisher information matrix}
\newac{ne}{NE}{Nash equilibrium}
\newac{mse}{MSE}{mean squared error}
\newac{toa}{TOA}{time-of-arrival}
\newac{snr}{SNR}{signal-to-noise ratio}
\newac{lan}{LAN}{local area network}
\newac{psd}{PSD}{positive semidefinite}
\newac{pd}{PD}{positive definite}
\newac{wrt}{w.r.t.}{with respect to}
\newac{lhs}{L.H.S.}{left hand side}
\newac{wp1}{w.p.1}{with probability 1}
\newac{kkt}{KKT}{Karush-Kuhn-Tucker}
\newac{wlog}{w.l.o.g.}{without loss of generality}
\newac{mle}{MLE}{maximum likelihood estimation}
\newac{gps}{GPS}{global positioning system}
\newac{rssi}{RSSI}{received signal strength indicator}
\newac{mimo}{MIMO}{multiple-input multiple-output}
\newac{csi}{CSI}{channel state information}
\newac{fdd}{FDD}{frequency division duplexing}
\newac{ms}{MS}{mobile station}
\newac{bs}{BS}{base station}
\newac{d2d}{D2D}{device-to-device}
\newac{slnr}{SLNR}{signal-to-interference-leakage-and-noise-ratio}
\newac{ula}{ULA}{uniform linear antenna array}
\newac{pas}{PAS}{power angular spectrum}
\newac{mmse}{MMSE}{minimum mean square error}
\newac{zf}{ZF}{zero-forcing}
\newac{rzf}{RZF}{regularized zero-forcing}
\newac{as}{AS}{angular spread}
\newac{aod}{AOD}{angle of departure}
\newac{iid}{i.i.d.}{independent and identically distributed} 
\newac{sinr}{SINR}{signal-to-interference-and-noise ratio}
\newac{tdd}{TDD}{time-division duplex}
\newac{rvq}{RVQ}{random vector quantization}
\newac{rhs}{R.H.S.}{right hand side}
\newac{mrc}{MRC}{maximum ratio combining}
\newac{cdf}{CDF}{cumulative distribution function}
\newac{a.s.}{a.s.}{almost surely}
\newac{los}{LOS}{line-of-sight}
\newac{jsdm}{JSDM}{joint spatial division and multiplexing}
\newac{map}{MAP}{maximum a posteriori}
\newac{klt}{KLT}{Karhunen-Lo\`eve Transform}
\newac{lbe}{LBE}{link bargaining equilibrium}
\newac{se}{SE}{Stackelberg equilibrium}
\newac{uav}{UAV}{unmanned aerial vehicle}
\newac{nlos}{NLOS}{non-line-of-sight}
\newac{pdf}{PDF}{probability density function}
\newac{em}{EM}{expectation-maximization}
\newac{knn}{KNN}{$k$-nearest neighbor}
\newac{svd}{SVD}{singular value decomposition}
\newac{nmf}{NMF}{non-negative matrix factorization}
\newac{umf}{UMF}{unimodality-constrained matrix factorization}
\newac{rmse}{RMSE}{rooted mean squared error}
\newac{olos}{OLOS}{obstructed line-of-sight}
\newac{mmw}{mmW}{millimeter wave}
\newac{ber}{BER}{bit error rate}
\newac{rss}{RSS}{received signal strength}
\newac{lp}{LP}{linear program}
\newac{ufw}{U-FW}{unimodal Frank-Wolfe}
\newac{utf}{UTF}{unimodality-constrained tensor factorization}
\newac{fw}{FW}{Frank-Wolfe}
\newac{iot}{IoT}{Internet-of-Things}
\newac{mae}{MAE}{mean absolute error}
\newac{crb}{CRB}{Cram\'er-Rao bound}
\newac{aoa}{AoA}{angle of arrival}
\newac{wcl}{WCL}{weighted centroid localization}
\newac{rb}{RB}{resource block}
\newac{aoi}{AoI}{age of information}
\newac{umi}{UMi}{Urban Micro}
\newac{ofdm}{OFDM}{orthogonal frequency division multiplexing}
\newac{iff}{iff}{if and only if}
\newac{qos}{QoS}{quality of service}
\newac{mpc}{MPC}{model predictive control}
\newac{vaoi}{VAoI}{version age of information}

%% file: figures/r1_paper_diagram.tex
\begin{tikzpicture}[
    >=Stealth,
    font=\footnotesize,
    box/.style={
        draw,
        rounded corners=1pt,
        align=center,
        minimum height=8.5mm,
        text width=2.15cm,
        inner sep=2.4pt
    },
    alg/.style={
        draw,
        align=center,
        minimum height=8mm,
        text width=2.0cm,
        inner sep=2.4pt
    },
    lab/.style={
        font=\scriptsize,
        fill=white,
        inner sep=1pt
    },
    arr/.style={->, thick}
]

\def\a{3.2}      
\def\b{1.4}      
\def\c{1.6}      
\def\dy{1}     

\node[box] (P1) at (0,0) {$\mathscr{P}1$\\Original\\bi-objective};
\node[box] (P2) at (\a,0) {$\mathscr{P}2$\\$\epsilon$-constraint\\reformulation};
\node[box] (P3) at ({2*\a},0) {$\mathscr{P}3$\\Explicit\\reformulation};

\coordinate (S) at ({2*\a+\b},0);

\node[box] (P31) at ({2*\a+\b+\c},\dy) {$\mathscr{P}\text{3-2}$\\Inner problem};
\node[box] (P4)  at ({2*\a+\b+\c+\a},\dy) {$\mathscr{P}4$\\Convex problem};
\node[alg] (A1)  at ({2*\a+\b+\c+2*\a},\dy) {Algorithm 1};

\node[box] (P32) at ({2*\a+\b+\c},-\dy) {$\mathscr{P}\text{3-1}$\\Outer problem};
\node[alg] (A2)  at ({2*\a+\b+\c+2*\a},-\dy) {Algorithm 2\\(MPComm)};

\draw[arr] (P1) -- node[lab, above] {(i)} (P2);
\draw[arr] (P2) -- node[lab, above] {(ii)} (P3);
\draw[-, thick] (P3.east) -- (S);

\draw[arr] (S) |- (P31.west);
\draw[arr] (S) |- (P32.west);
\node[lab, right] at ($(S)+(0.1,0)$) {(iii)};

\draw[arr] (P31) -- node[pos=0.5, lab, above] {(iv)} (P4);
\draw[arr] (P4)  -- node[pos=0.5, lab, above] {(v)} (A1);

\draw[arr] (P32) -- node[pos=0.5, lab, above] {(vi)} (A2);
\draw[arr] (A1.south) -- node[lab, left] {(vii)} (A2.north);
\end{tikzpicture}

%% file: figures/two_layer.tex
\begin{tikzpicture}[
    node distance=1.5cm,
    box/.style={rectangle, draw, minimum width=2.5cm, minimum height=1cm, align=center},
    dashed box/.style={rectangle, draw, dashed, minimum width=8cm, minimum height=3cm}
]

\node [box, thick, rounded corners=1.5pt, rectangle, minimum width=1.5cm, minimum height=1cm] (spi) {Inner Layer\\ ($\mathscr{P}\text{3-2}$, Sec.~\ref{subsec:physical_layer_control})};
\node [box, thick, rounded corners=1.5pt, rectangle, right=1.8cm of spi, minimum width=2cm, minimum height=1cm] (lagrange) {Outer Layer\\
($\mathscr{P}\text{3-1}$, Sec.~\ref{subsec:graph})};
\node[right=2.0cm of lagrange] (lam) {};

\draw[->, thick] (spi) -- node[above] {$\{\boldsymbol{\mathcal{A}}_i, \boldsymbol{\mathcal{P}}_i\}$} (lagrange);
\draw[->, thick] (lagrange) -- ++(0, -1.) -| node[near start, below] {$\boldsymbol{t} = \{t_0, t_1, \ldots, t_I\}$} (spi);
\draw[->, thick] (lagrange) -- node[above] {$\boldsymbol{t}^*,\boldsymbol{\mathcal{A}}^*, \boldsymbol{\mathcal{P}}^*$} coordinate[midway] (mid-cd) (lam);

\end{tikzpicture}

%% file: figures/graph_construction.tex
\begin{tikzpicture}[scale=1]
\node[draw, circle, minimum size=0.2cm, line width=1.5] (n1) at (0,0) {};
\node[below=0.5em] at (n1) {$1$};

\node [draw, circle, right=0.8cm of n1, minimum size=0.2cm, line width=1.5] (n2) {};
\node[below=0.5em] at (n2) {$2$};

\node [draw, circle, right=0.8cm of n2, minimum size=0.2cm, line width=1.5] (n3) {};
\node[below=0.5em] at (n3) {$3$};

\fill (3.1,0) circle (1.2pt);
\fill (3.4,0) circle (1.2pt);
\fill (3.7,0) circle (1.2pt);

\node [draw, circle, right=1.8cm of n3, minimum size=0.2cm, line width=1.5] (n4) {};
\node[below=0.5em] at (n4) {$\bar{\tau}\!+\!1$};

\node [draw, circle, right=0.8cm of n4, minimum size=0.2cm, line width=1.5] (n5) {};
\node[below=0.5em] at (n5) {$\bar{\tau}\!+\!2$};

\fill (6.5,0) circle (1.2pt);
\fill (6.8,0) circle (1.2pt);
\fill (7.1,0) circle (1.2pt);

\node [draw, circle, right=1.8cm of n5, minimum size=0.2cm, line width=1.5] (n6) {};
\node[below=0.5em] at (n6) {\small $T\!+\!1$};

\draw[->, thick] (n1) to[bend left=20] (n2);
\draw[->, thick] (n1) to[bend left=24] (n3);
\draw[->, thick] (n1) to[bend left=26] (n4);
\draw[->, thick, dashed] (n1) to[bend left=26] node[pos=0.85, sloped, above] {\small infeasible} (n5);
\draw[->, thick, dashed] (n1) to[bend left=28] (n6);
\draw[->, thick] (n2) to[bend right=20] (n3);
\draw[->, thick] (n2) to[bend right=24] (n4);
\draw[->, thick] (n2) to[bend right=26] (n5);

\end{tikzpicture}